\documentclass[a4paper,hidelipics]{lipics-v2021}
\newif\iflong
\newif\ifshort

\longtrue

\iflong
\else
\shorttrue
\fi

\newcommand\blfootnote[1]{  \begingroup
  \renewcommand\thefootnote{}\footnote{#1}  \addtocounter{footnote}{-1}  \endgroup
}

\hideLIPIcs

\newcommand{\MEP}{\lambda}

\usepackage{microtype}
\usepackage{graphicx}
\usepackage{floatrow}
\usepackage{enumitem}
\usepackage{url}            \usepackage{booktabs}       \usepackage{amsfonts}       \usepackage{nicefrac}       \usepackage{microtype}      \usepackage{verbatim}

\usepackage[ruled,linesnumbered]{algorithm2e}

\renewenvironment{claimproof}[1][\proofname]{\proof[#1]}{\endproof}

\nolinenumbers
\SetAlFnt{\small}
\SetAlCapFnt{\small}
\SetAlCapNameFnt{\small}
\SetAlCapHSkip{0pt}
\IncMargin{-\parindent}

\setlist[itemize]{leftmargin=*}
\setlist{nolistsep}
\usepackage{booktabs} 
\usepackage{hyperref}

\usepackage{boxedminipage}
\usepackage{xspace}

\usepackage{amsmath, amssymb}
\usepackage{amsthm}

\newcommand{\save}[1]{}

\usepackage[obeyFinal,colorinlistoftodos,textsize=tiny]{todonotes}

\theoremstyle{plain}
\newtheorem{THE}{Theorem}

\newtheorem{LEM}[THE]{Lemma}
\newtheorem{COR}[THE]{Corollary}
\newtheorem*{THE*}{Theorem}
\newtheorem*{fact*}{Fact}

\theoremstyle{definition}

\newtheorem*{DEF*}{Definition}

\newtheorem{OBS}[THE]{Observation}

\newtheorem{CLM}{Claim}

\newcommand{\SB}{\{\,}
\newcommand{\SM}{\;{|}\;}
\newcommand{\SE}{\,\}}
\newcommand{\AAA}{\mathcal{A}}
\newcommand{\BBB}{\mathcal{B}}

\newcommand{\CCC}{\mathcal{C}}

\newcommand{\SSS}{\mathcal{S}}

\newcommand{\III}{\mathcal{I}}
\newcommand{\FFF}{\mathcal{F}}

\newcommand{\Nat}{\mathbb{N}}

\newcommand{\cc}[1]{{\mbox{\textnormal{\textsf{#1}}}}\xspace}

\newcommand{\NP}{\cc{NP}}

\newcommand{\FPT}{\cc{FPT}}
\newcommand{\XP}{\cc{XP}}
\newcommand{\Weft}{{\cc{W}}}
\newcommand{\W}[1]{{\Weft}{{\normalfont{[#1]}}}}
\newcommand{\paraNP}{\cc{paraNP}}

\newcommand{\hy}{\hbox{-}\nobreak\hskip0pt}

\newcommand{\bigoh}{\mathcal{O}}

\newcommand{\probfont}[1]{\textnormal{\textsc{#1}}}

\newcommand{\blank}{{\small\square}}

\newcommand{\DIAMC}{\probfont{Diam-Cluster}}
\newcommand{\DIAMCq}{\probfont{Diam-Cluster-Completion}}
\newcommand{\RADC}{\probfont{Rad-Cluster}}
\newcommand{\RADCq}{\probfont{Rad-Cluster-Completion}}

\newcommand{\threshold}[1]{#1-saturated}

\newcommand{\DCLUS}{DIAM-Cluster}
\newcommand{\RCLUS}{RAD-Cluster}
\newcommand{\MEv}{\textup{ME}}

\newcommand{\yes}{\textsc{Yes}}

\def\ie{{i.e.}}
\def\eg{{e.g.}}
\def\etal{{et al.}}

\newcommand{\YES}{\yes}
\newcommand{\HSET}{\Delta}

\newcommand{\HDIST}{\delta}
\newcommand{\DIAM}{\gamma}
\newcommand{\DCOOR}{Z}
\newcommand{\HN}[2]{N_{#2}({#1})}
\newcommand{\inst}{\mathcal{I}}
\newcommand{\compG}{G}

\title{Finding a Cluster in Incomplete Data}

\author{Eduard Eiben}{Department of Computer Science, Royal Holloway, University of London, Egham, UK}{eduard.eiben@gmail.com}{https://orcid.org/0000-0003-2628-3435}{}

\author{Robert Ganian}{Algorithms and Complexity Group, TU Wien, Vienna, Austria}{rganian@gmail.com}{https://orcid.org/0000-0002-7762-8045}{Project No. Y1329 of the Austrian Science Fund (FWF)}

\author{Iyad Kanj}{School of Computing, DePaul University, Chicago, USA}{ikanj@cdm.depaul.edu}{}{}

\author{Sebastian Ordyniak}{University of Leeds, School of Computing, Leeds, UK}{sordyniak@gmail.com}{https://orcid.org/0000-0003-1935-651X}{Project EP/V00252X/1 of the Engineering and Physical Sciences Research Council (EPSRC)}

\author{Stefan Szeider}{Algorithms and Complexity Group, TU Wien, Vienna, Austria}{sz@ac.tuwien.ac.at}{https://orcid.org/0000-0001-8994-1656}{Project No. P32441 of the Austrian Science Fund (FWF) and Project No. ICT19-065 of the  Vienna Science and Technology Fund (WWTF)}

\authorrunning{E.\ Eiben, R.\ Ganian, I.\ Kanj, S.\ Ordyniak, S.\ Szeider} %

\Copyright{Eduard Eiben, Robert Ganian, Iyad Kanj, Sebastian Ordyniak, Stefan Szeider} %

\ccsdesc[300]{Theory of computation~Parameterized complexity and exact algorithms}
\keywords{Parameterized complexity, incomplete data, clustering}

\begin{document}

\maketitle

\ifshort
\vspace{1cm}
\fi
\begin{abstract}
We study two variants of the fundamental problem of finding a cluster in incomplete data. In the problems under consideration, we are given a multiset of incomplete $d$-dimensional
vectors over the binary domain and integers $k$ and~$r$, and the goal is to complete the missing
vector entries so that the multiset of complete vectors either contains (i) a cluster of $k$ vectors of radius at most $r$, or (ii) a cluster of $k$ vectors of diameter at most $r$.  We give tight characterizations of the parameterized
complexity of the problems under consideration with respect to the parameters $k$,
$r$, and a third parameter that captures the missing vector entries. 
\ifshort\blfootnote{\noindent
\emph{Statements whose proofs
            or details are provided in the appendix are marked with
            $\spadesuit$.}}
 	\fi
 \end{abstract}

\section{Introduction}\label{sec:intro}

 We consider two formulations of the fundamental problem of finding a sufficiently large cluster in incomplete data~\cite{boucherma,BulteauS20,GrammNiedermeierRossmanith03,Li}.
  In the setting under consideration, the input is a multiset $M$ of $d$-dimensional Boolean vectors---regarded as the rows of a matrix,
 some of whose entries might be missing---and two parameters $k, r \in \Nat$. In the first problem under consideration, referred to as \DIAMCq{}, the goal is to decide whether there is a completion of $M$ that admits a multiset of $k$ vectors ($k$-cluster) of diameter at most $r$; that is, a $k$-cluster such that the Hamming distance between any two cluster-vectors is at most $r$. In the second problem, referred to as \RADCq{}, the goal is to decide whether there is a completion of $M$ that admits a $k$-cluster of radius at most $r$; that is, a $k$-cluster such that there is a \emph{center} vector $\vec{s}\in \{0, 1\}^d$ with Hamming distance at most $r$ to each cluster-vector.

The cluster-diameter and cluster-radius are among the most widely used measures for intra-cluster similarity~\cite{charikar,frieze,normclustering,lingashammingcenter,lingasapxclustering,gonzalez,GrammNiedermeierRossmanith03}.
 Our interest in studying these problems stems from the recent relevant research within the theory community~\cite{EibenGKOS21,icml,hermelin,KoanaFN20,KoanaFN21}, as well as the ubiquitous presence of incomplete data in relevant areas, such as recommender systems, machine learning, computer vision, and data science~\cite{balzano,nips2016-incompleteclustering,ev13,icdm}.

We study the parameterized complexity of the above problems
with respect to the two parameters $k$, $r$, and a third
parameter that captures the occurrence of missing vector
entries. Naturally, parameterizing by the number of missing entries
alone is not desirable since one would expect that number
to be rather large. In their recent related works on
clustering problems, Koana, Froese and
Niedermeier~\cite{KoanaFN20,KoanaFN21} restricted the occurrence of
missing entries by using the maximum number of missing entries per row
as the parameter. Another parameter for restricting the occurrence of missing entries is the minimum number of vectors plus coordinates
needed to ``cover'' all missing entries, which was proposed and used by Eiben \etal~\cite{EibenGKOS21} and Ganian \etal~\cite{icml}, who
studied various data completion and clustering problems.  In
this paper, we propose and use a parameter that unifies and subsumes
both previous parameterizations:
the ``\emph{deletion distance
  to near-completion}'', denoted $\lambda(M)$, which
is the minimum integer $p$ such that at most $p$ vectors can be removed
from $M$ so that every remaining vector contains at most $p$ missing
entries. Clearly, the parameter $\lambda(M)$ is computable in polynomial
time and is not larger than any of (and hence subsumes) the two parameters
considered by Koana \etal~\cite{KoanaFN20,KoanaFN21}, Eiben
\etal~\cite{EibenGKOS21}, and Ganian \etal~\cite{icml}.  

\smallskip

\noindent \textbf{Results and Techniques.}\quad
We perform an in-depth analysis of the two considered data
completion problems w.r.t.~the aforementioned parameterizations. We
obtain results that provide a nearly complete complexity landscape of these problems. An overview of our results is provided in Table~\ref{tab:results-combined}. 
As a byproduct, our results establish that both problems under consideration are fixed-parameter tractable parameterized by $k+r$ when the data is complete, which answers an open question in the literature~\cite{boucherma,BulteauS20}.
 
\begin{table*}[htbp]
 
  \begin{center}
    \begin{tabular}{@{}l@{\quad}c@{\quad}c@{\quad}c@{\quad}c@{\quad}c@{\quad}c@{}}\toprule
      ~& $k$ & $r$ & $k+r$ & $k+\MEP$ & $r + \MEP$ & $k+r+\MEP$ \\ \midrule

    \textsc{Diam-Cluster-C.}
            & \W{1}-h/\XP{} & \paraNP{}\hy c &  \W{1}-h/\XP{} & \W{1}-h/\XP{} & \FPT{}& \FPT{}\\

      \textsc{Rad-Cluster-C.}
      & \W{1}-h/\XP{} & \paraNP{}\hy c &  \W{1}-h/\XP{} &
                                                                    ?/\XP{} & ?/\XP{} & \FPT{} \\  
                                               
     \bottomrule
    \end{tabular}
  \end{center}
  \vspace{-0.5cm}
  \caption{\small Overview of the results obtained in this paper.}

  \label{tab:results-combined}
\end{table*}

We summarize the new results obtained in this paper below.
\smallskip

\noindent \textbf{(1)}\quad
We show that \DIAMCq{} is fixed-parameter tractable (\FPT{}) parameterized by $r+\MEP(M)$ (Theorem~\ref{thm:diamfpt}). The significance of the above result is in removing the dependency on the cluster size $k$ in the running time of the algorithm, thus showing that finding a large cluster in incomplete data can be feasible when both the cluster diameter and the parameter $\MEP(M)$ are small.  This result is the pinnacle of our technical contributions and relies on two ingredients: a fixed-parameter algorithm for the same problem parameterized by $k+r+\MEP(M)$ (Theorem~\ref{thm:DIAM_k_r_comb}), which is then used as a subroutine in the main algorithm,
and a new technique that we dub \emph{iterative sunflower harvesting}.  Crucial to this new technique is a general structural lemma, allowing us to represent a family of sets in a succinct manner in terms of sunflower cores, which we believe to be interesting in its own right.
\iflong
We note that the use of sunflowers to obtain a succinct set representation (leading to a kernel) is not uncommon in such settings. For instance, a similar set representation was also used in a previous work by Marx~\cite{marx} (see also Kratsch, Marx and Wahlstr\"om~\cite{kmw}) to obtain \FPT{} and kernelization results, albeit in the context of weighted CSPs. Koana \etal~\cite{KoanaFN20, KoanaFN21} also use the Sunflower Lemma in their results. \fi
\ifshort
We note that the use of sunflowers to obtain a succinct set representation (leading to a kernel) is not uncommon in such settings~\cite{KoanaFN20,KoanaFN21,kmw,marx}.
\fi
What makes the sunflower harvesting technique novel is that it allows us to (1) show that each solution can be covered by a small number of sunflowers, and to (2) iteratively “harvest” these sunflowers to obtain a solution in boundedly-many (in the parameter) branching steps. 

\smallskip
\noindent \textbf{(2)}\quad We give an \XP-algorithm for \DIAMCq{} parameterized by $k$ alone (Theorem~\ref{the:diam-xp-k}). Together with Theorem~\ref{the:diam-w1-k} (showing the \W{1}\hy hardness of \DIAMC{} w.r.t.\ $k$)  and Theorem~\ref{thm:diamrad-Whard-k}~(showing the \W{1}\hy hardness for \DIAMCq{} w.r.t.\ $k$ even for $r=0$), this gives a complete
complexity landscape for \DIAMCq{} parameterized by any combination
of the parameters~$k$, $r$, and $\MEP(M)$.

\smallskip
\noindent \textbf{(3)}\quad We show that \RADCq{} is \FPT{} parameterized by $k+r+\MEP(M)$ (Theorem~\ref{thm:ANY_k_r_comb}); this result answers open questions in the literature~\cite{boucherma,BulteauS20}, which asked about the fixed-parameter tractability of the easier complete version of the problem (i.e., when $\MEP(M)=0$). 

\smallskip
\noindent \textbf{(4)}\quad  We provide an \XP-algorithm for \RADCq\ parameterized by $r+\MEP(M)$ in which the degree of the polynomial in the runtime has only a logarithmic dependence on~$r$ (Theorem~\ref{thm:MAX-ANY-r+comb}). We remark that the problem is in \XP\ parameterized by $k$ alone (Observation~\ref{obs:XPk}).

\smallskip
\noindent \textbf{(5)}\quad We provide an accompanying \W{1}-hardness result for \RADCq\ that rules out its fixed-parameter tractability when parameterized by $k+r$ (Theorem~\ref{thm:diamrad-Whard-k}). Since the problem is \NP-hard for fixed $r$ (as also follows from Theorem~\ref{thm:diamrad-Whard-k}), this leaves only two questions open for the considered parameterizations: whether \RADCq\ is \FPT\ when parameterized by either $k+\MEP(M)$ or by $r+\MEP(M)$. 

\smallskip
\noindent \textbf{(6)}\quad We give an \FPT{}-approximation scheme for the optimization version (w.r.t.~the cluster size) of \RADCq{}.

\smallskip
\noindent \textbf{Related Work.}\quad
The \RADC{} problem (i.e., \RADCq{}  for complete data) and variants of it were studied as early as the 1980's, albeit under different names. Dyer and Frieze presented a heuristic algorithm for approximating a variant of \RADC{}, referred to as the {\sc $p$-Center} problem, where the goal is to compute $p \in \Nat$ clusters, each of radius at most $r$, that contain all vectors of $M$; hence, \RADC{} corresponds to the case of {\sc $p$-Center} where $p=1$ and $k=|M|$ (i.e., when the cluster contains all vectors in $M$). Cabello \etal~\cite{cabello} studied the parameterized complexity of the geometric \textsc{$p$-Center Problem} in $\mathbb{R}^d$.

Frances and Litman~\cite{litman} studied the complexity of \RADC{} with $k=|M|$, in the context of computing the radius of a binary code; they referred to it as the {\sc Covering Radius} problem and showed it to be \NP-hard. G\c{a}sieniec \etal~\cite{lingashammingcenter,lingasapxclustering} studied (the
optimization versions of) \RADC{} and \DIAMC{} with $k=|M|$ and obtained polynomial-time algorithms
as well as lower bounds for a number of cases. They also obtained $2$-approximation algorithms for these problems by extending an earlier algorithm by Gonzalez~\cite{gonzalez}.

The \RADC{} problem restricted to the subcase of $k=|M|$ was also extensively studied under the nomenclature {\sc Closest String}. Li \etal~\cite{Li} showed that the problem admits a polynomial time approximation scheme if the goal is to minimize $r$.
Gramm \etal~\cite{GrammNiedermeierRossmanith03} studied
{\sc Closest String} from the parameterized complexity perspective and showed it to be fixed-parameter tractable parameterized by $r$. Following this naming convention, Boucher and Ma~\cite{boucherma}, and Bulteau and Schmid~\cite{BulteauS20} studied the parameterized complexity of \RADC{} under the nomenclature {\sc Closest String with Outliers}. They considered several parameters, including some of the parameters under consideration in this paper. Notably, the restriction of our fixed-parameter algorithm for \RADC{} parameterized by $k+r+\lambda$ to the subcase where $\lambda=0$ answers an open question in~\cite[see, \eg, Table 1]{boucherma}.
Moreover, our \XP\ algorithm for \RADCq\ provided in
Theorem~\ref{thm:MAX-ANY-r+comb} that has a run-time in which the degree
of the polynomial has only a logarithmic dependence on $r$,
immediately implies an algorithm of the same running time for
\textsc{Closest String with Outliers}, as a special case.
 
For incomplete data, Hermelin and Rozenberg~\cite{hermelin} studied the parameterized complexity of \RADCq{} for $k=|M|$ under the nomenclature {\sc Closest String with Wildcards} problem, with respect to several parameterizations. Very recently, Koana \etal~\cite{KoanaFN20} revisited the earlier work of Hermelin and Rozenberg~\cite{hermelin} and obtained, among other results, a fixed-parameter algorithm for the problem parameterized by $r$ plus the maximum number of missing entries per row.
Even more recently, the same group~\cite{KoanaFN21} also studied a problem related to \DIAMCq{} for $k=|M|$. They obtain a classical-complexity classification w.r.t.\ constant lower and upper bounds on the diameter and the maximum number of missing entries per row.

 \section{Preliminaries}\label{section:prelims}

 \ifshort
We assume basic familiarity with parameterized complexity, including the classes \W{1}, \FPT, \XP, as well as Turing kernelization and \FPT-approximation schemes~\cite{DowneyFellows13,CyganFKLMPPS15,Marx08FPTAS}. ($\spadesuit$)
\fi

\newcommand{\compGq}{\compG}

\newcommand{\HSETq}{\HSET}
\newcommand{\HDISTq}{\HDIST}
\newcommand{\DIAMq}{\DIAM}
\newcommand{\MDIAMq}{\DIAM_{\max}}
\newcommand{\HNq}{\HN}

\smallskip
\noindent \textbf{Vector Terminology.}\quad
Let $\vec{a}$ and $\vec{b}$ be two vectors in $\{0,1,\blank\}^d$, where $\blank$ is used to represent coordinates whose value is unknown (\ie, missing entries). We denote
by $\HSET(\vec{a},\vec{b})$ the set of coordinates in which
$\vec{a}$ and $\vec{b}$ are guaranteed to differ, \ie,
$\HSET(\vec{a},\vec{b})=\SB i
\SM (\vec{a}[i]=1 \wedge \vec{b}[i]=0)\vee (\vec{a}[i]=0 \wedge \vec{b}[i]=1) \SE$, and we denote by
$\HDIST(\vec{a},\vec{b})$ the \emph{Hamming distance} between
$\vec{a}$ and~$\vec{b}$ measured only between known entries, \ie,
$|\HSET(\vec{a},\vec{b})|$. Moreover, for a subset $D' \subseteq [d]$ of coordinates, where $[d]=\{1,\dots,d\}$, we denote by $\vec{a}[D']$ the vector $\vec{a}$
restricted to the coordinates in $D'$.

There is a one-to-one correspondence between vectors in
$\{0,1\}^d$ and
subsets of coordinates, \ie, for every vector, we can associate the unique
subset of coordinates containing all its one-coordinates and
vice-versa.
It will be useful to represent a vector by the set of coordinates where the vector has the value $1$.
We introduce the following notation for vectors to switch between their set-representation
and vector-representation. We denote by $\HSET(\vec{a})$ the set
$\HSET(\vec{0},\vec{a})$.
We extend this notation to sets of
vectors as follows: for a set $N$ of
vectors in $\{0,1,\blank\}^d$, we denote by $\HSET(N)$ the set $\SB\HSET(\vec{v}) \SM \vec{v} \in N\SE$.
We say that a vector $\vec{a}\in \{0,1,\blank\}^d$
is a \emph{$t$-vector} if $|\HSET(\vec{a})|=t$ and
we say that $\vec{a}$ \emph{contains} a subset $S$ of coordinates if $S \subseteq \HSET(\vec{a})$.

\iflong
\enlargethispage{4mm}
\fi
We say that a multiset\footnote{We remark that, in the interest of brevity and when clear from context, we will sometimes use standard set notation such as $A\subseteq B$ in conjunction with multisets.}
$M^*\subseteq \{0,1\}^d$ is a \emph{completion} of a multiset $M\subseteq \{0,1,\blank\}^d$ if there is a bijection $\alpha: M \rightarrow M^*$ such that for all $\vec{a}\in M$ and all $i\in [d]$ it holds that either $\vec{a}[i]=\blank$ or $\alpha(\vec{a})[i]=\vec{a}[i]$.
For a multiset $M$ of vectors over $\{0,1,\blank\}^d$, we let \emph{the deletion distance to near-completion}, $\lambda(M)$, denote the minimum integer such that there exists a subset $D_M\subseteq M$ with the following properties: (a) $|D_M|\leq \lambda(M)$, and (b) every vector in $M\setminus D_M$ contains at most $\lambda(M)$ missing entries. We call $D_M$ the \emph{deletion (multi-)set}, and observe that $\lambda(M)$ along with a corresponding deletion set can be trivially computed from $M$ in linear time.

A \emph{sunflower} in a set family $\FFF$ is a subset $\FFF' \subseteq \FFF$ such that all pairs of elements in $\FFF'$ have the same intersection.
\iflong The following lemma captures a useful observation that is used in several of our proofs.

\begin{LEM}[\cite{EibenGKOS21}]\label{lem:CLUS-SF-BASIC-A}
    Let $t, r \in \Nat$. Let $N \subseteq \{0,1\}^d$ be a set
    of $t$-vectors such that $\FFF:=\HSET(N)$ is a sunflower with core
    $C$. If $|N|>r$, then
    for every vector $\vec{a} \in \{0,1\}^d$ with $|\HSET(\vec{a})|\leq
    r$, $N$ contains a vector that has maximum distance to $\vec{a}$ among all
    $t$-vectors that contain $C$.
  \end{LEM}
\fi
  We will say that a multiset $P$ is a
 \emph{\DCLUS{}} (or $|P|$-\emph{\DCLUS{}}) if $\HDIST(\vec{p},\vec{q})\leq r$ for every pair $\vec{p},\vec{q}\in P$. Similarly, $P$ is a \emph{\RCLUS{}} (or $|P|$-\emph{\RCLUS{}}) if there exists a vector $\vec{c}\in \{0,1\}^d$ such that $\HDIST(\vec{c},\vec{p})\leq r$ for every $\vec{p}\in P$.

\iflong
\smallskip
\noindent \textbf{Parameterized Complexity.}\quad
In parameterized complexity~\cite{FlumGrohe06,DowneyFellows13,CyganFKLMPPS15},
the complexity of a problem is studied not only with respect to the
input size, but also with respect to some problem parameter(s). The
core idea behind parameterized complexity is that the combinatorial
explosion resulting from the \NP-hardness of a problem can sometimes
be confined to certain structural parameters that are small in
practical settings. We now proceed to the formal definitions.

A {\it parameterized problem} $Q$ is a subset of $\Omega^* \times
\mathbb{N}$, where $\Omega$ is a fixed alphabet. Each instance of $Q$ is a pair $(I, \kappa)$, where $\kappa \in \Nat$ is called the {\it
parameter}. A parameterized problem $Q$ is
{\it fixed-parameter tractable} (\FPT)~\cite{FlumGrohe06,DowneyFellows13,CyganFKLMPPS15}, if there is an
algorithm, called an {\em \FPT-algorithm},  that decides whether an input $(I, \kappa)$
is a member of $Q$ in time $f(\kappa) \cdot |I|^{\bigoh(1)}$, where $f$ is a computable function and $|I|$ is the input instance size.  The class \FPT{} denotes the class of all fixed-parameter
tractable parameterized problems.

A parameterized problem $Q$
is {\it \FPT-reducible} to a parameterized problem $Q'$ if there is
an algorithm, called an \emph{\FPT-reduction}, that transforms each instance $(I, \kappa)$ of $Q$
into an instance $(I', \kappa')$ of
$Q'$ in time $f(\kappa)\cdot |I|^{\bigoh(1)}$, such that $\kappa' \leq g(\kappa)$ and $(I, \kappa) \in Q$ if and
only if $(I', \kappa') \in Q'$, where $f$ and $g$ are computable
functions. By \emph{\FPT-time}, we denote time of the form $f(\kappa)\cdot |I|^{\bigoh(1)}$, where $f$ is a computable function.
Based on the notion of \FPT-reducibility, a hierarchy of
parameterized complexity, {\it the \cc{W}-hierarchy} $=\bigcup_{t
\geq 0} \W{t}$, where $\W{t} \subseteq \W{t+1}$ for all $t \geq 0$, has
been introduced, in which the $0$-th level \W{0} is the class {\it
\FPT}. The notions of hardness and completeness have been defined for each level
\W{$i$} of the \cc{W}-hierarchy for $i \geq 1$ \cite{DowneyFellows13,CyganFKLMPPS15}. It is commonly believed that $\W{1} \neq \FPT$ (see \cite{DowneyFellows13,CyganFKLMPPS15}). The
\W{1}-hardness has served as the main working hypothesis of fixed-parameter
intractability. The class \XP{} contains parameterized problems that can be solved in time  $\bigoh(|I|^{f(\kappa)})$, where $f$ is a computable function; it
contains the class \W{t}, for $t \geq 0$, and every problem in \XP{} is polynomial-time solvable when the parameters are bounded by a constant.
The class \paraNP{} is the class of parameterized problems that can be solved by non-deterministic algorithms in time $f(\kappa)\cdot |I|^{\bigoh(1)}$, where $f$ is a computable function.
A problem is \emph{\paraNP{}-hard} if it is \NP-hard for a constant value of the parameter~\cite{FlumGrohe06}.

A parameterized problem is {\em kernelizable}
if there exists a polynomial-time reduction that maps an instance $(I, \kappa)$ of
the problem to another instance $(I', \kappa')$ such that (1) $|I'| \leq f(\kappa)$ and $\kappa' \leq f(\kappa)$, where $f$ is a computable function, and (2) $(I,\kappa)$ is a \yes-instance
of the problem if and only if $(I',\kappa')$ is. The instance
$(I',\kappa')$ is called the {\em kernel} of~$I$. It is well known that a
decidable problem is \FPT{} if and only if it is
kernelizable~\cite{DowneyFellows13}.
A \emph{polynomial kernel} is a kernel whose size can be bounded by a
polynomial in the parameter.

A \emph{Turing kernelization} for a parameterized problem $Q$ is an algorithm that, provided
with access to an oracle for $Q$, decides in polynomial time whether or not an input $(I, \kappa)$
is a YES-instance of $Q$. During its
computation, the algorithm can produce polynomially-many oracle queries on input of size $f(\kappa)$, for some computable function $f$. The function $f$ is referred to as the \emph{kernel size}. 
An \emph{\FPT{} approximation scheme} for a parameterized
maximization problem is an algorithm that takes as input an instance
$(I, \kappa)$ of $Q$ and $0<\epsilon \leq
1$, and produces in \FPT{}-time in $\kappa+1/\epsilon$ a solution to $I$ that is an $(1-\epsilon)$-approximation for the instance.  

\fi
\section{Finding a DIAM-Cluster in Incomplete Data}
In this section, we present our results for \DIAMCq{}. Our main algorithmic results
are that \DIAMCq{} is \FPT{} parameterized by
$r+\MEP(M)$ and is in \XP{} parameterized by $k$ alone. Together with
Theorem~\ref{the:diam-w1-k} (showing the \W{1}\hy hardness of \DIAMC{} parameterized by $k$)
and Theorem~\ref{thm:diamrad-Whard-k} (showing the \W{1}\hy hardness of
\DIAMCq{} parameterized by $k$ even for $r=0$), this gives a complete
complexity landscape for \DIAMCq{} parameterized by any combination
of the parameters $k$, $r$, and $\MEP(M)$.

\subsection{\DIAMCq{} Parameterized by $k+r+\MEP(M)$}
\label{sssec:diam-k-r-comb}
We start by showing that \DIAMCq{} parameterized by $k+r+\MEP(M)$
is \FPT{}. We will later show a stronger
result, namely that the same result already holds if we only
parameterize by $r+\MEP(M)$. Showing the weaker result here is
important for the following reasons: (1) we use the algorithm presented here as a subroutine in our result
for the parameterization $r+\MEP(M)$, (2) the techniques developed here can also be employed for \RADCq, and
(3) we obtain a Turing kernel of size polynomial in~$k$.

\ifshort
The main approach behind the Turing kernel is to guess two vectors of maximum distance in
the desired cluster. This will allow us to pre-process the instance such that if the
resulting instance contains too many vectors, then it has a
solution. Note that this approach only works for the case that a
solution contains at least two vectors from $M\setminus D_M$ (recall that $D_M$ denotes the deletion set);
otherwise, we can guess the at most one vector from $M\setminus D_M$ that
is in the solution and remove all the other vectors from $M\setminus
D_M$. Therefore, in all cases, we end up with a reduced instance with
boundedly many vectors and we will then show that we can remove all
but boundedly many coordinates while preserving solutions.
\fi
\iflong
The following lemma shows how this
can be used to pre-process the instance (by identifying vectors that
cannot be in a cluster with the two vectors of maximum distance). Note
that this idea and the following lemma only deal with the case that the
solution cluster uses at least two vectors from $M\setminus D_M$,
where $D_M$ is a deletion set such that every vector in $M\setminus
D_M$ has at most $\MEP(M)$ missing entries. For a vector $\vec{v}$, we
denote by $\MEv(\vec{v})$ the set of all missing coordinates of $\vec{v}$.

\iflong \begin{LEM} \fi \ifshort \begin{LEM}[$\spadesuit$] \fi\label{lem:krcom-prune}
  Let $\vec{v}$ and $\vec{u}$ be two vectors in $M\setminus D_M$ and
  let $t$ be an integer with $\HDIST(\vec{v},\vec{u}) \leq t \leq
  \HDIST(\vec{v},\vec{u})+2|D_M|$. Then any vector $\vec{m}$ satisfying
  $|\HSET(\vec{v},\vec{m})|>t$, $|\HSET(\vec{u},\vec{m})|>t$, or
  $|\HSET(\vec{v},\vec{m})\cap\HSET(\vec{u},\vec{m})|>t/2$ cannot be
  contained in a \DCLUS{}/\RCLUS{} containing $\vec{v}$ and
  $\vec{u}$, and in which $\vec{v}$ and
  $\vec{u}$ are two vectors of maximum distance $t$.
\end{LEM}
\iflong \begin{proof}
  Clearly, the lemma holds for any vector $\vec{m}$ satisfying either
  $|\HSET(\vec{v},\vec{m})|>t$ or $|\HSET(\vec{v},\vec{m})|>t$,
  since we are only looking for clusters where $\vec{v}$ and $\vec{u}$
  are of distance $t$
    to one another and have maximum distance within
  the cluster. Now let $\vec{m}$ be a vector satisfying
  $|\HSET(\vec{v},\vec{m})\cap\HSET(\vec{u},\vec{m})|>t/2$.
  Since the completions of $\vec{v}$ and $\vec{u}$ in a solution differ in exactly $t$ coordinates in
  $\HSET(\vec{v},\vec{u})\cup \MEv(\vec{v})\cup\MEv(\vec{u})$, it follows that $\vec{m}$ differs
  in at least $t/2$ of those coordinates from either $\vec{v}$ or
  $\vec{u}$. It follows that $\vec{m}$ differs in more than $t$ coordinates
  from either $\vec{u}$ or $\vec{v}$, and hence cannot be part of a
  cluster with $\vec{v}$ and $\vec{u}$ of cluster diameter at most $t$.
\end{proof} \fi

The following lemma shows that the instance obtained after
identifying two vectors of maximum distance and removing the vectors
identified by the previous lemma, cannot contain too many vectors since
otherwise, it would be a \YES{}\hy instance.
\iflong \begin{LEM} \fi \ifshort \begin{LEM}[$\spadesuit$] \fi\label{lem:krcom-kernel}
  Let $\vec{v}$ and $\vec{u}$ be two vectors in $M\setminus D_M$ and
  let $t$ be an integer with $\HDIST(\vec{v},\vec{u}) \leq t \leq
  \HDIST(\vec{v},\vec{u})+2|D_M|$. Let $M'$ be the set of all vectors
  in $M\setminus (D_M\cup \{\vec{v},\vec{u}\})$ obtained after removing all vectors $\vec{m}$
  for which either
  $|\HSET(\vec{v},\vec{m})|)>t$, $|\HSET(\vec{u},\vec{m})|)>t$, or
  $|\HSET(\vec{v},\vec{m})\cap\HSET(\vec{u},\vec{m})|>t/2$. If
  $|M'|>k3^{2|D_M|+t}$, then $M'$ contains a
  \DCLUS{} of size at least $k$ and diameter $t$, which is
  also a \RCLUS{} of radius $t/2$.
\end{LEM}
\iflong \begin{proof}
  Assume that there are at least $k3^{2|D_M|+t}$
  vectors $\vec{m}$ in $M'$. Then there is a set $S\subseteq M'$ of
  size at least $k$ that agrees on
  all the coordinates in $\HSET(\vec{v},\vec{u})\cup \MEv(\vec{v})\cup
  \MEv(\vec{u})$; this is true since
  $|\HSET(\vec{v},\vec{u})\cup \MEv(\vec{v})\cup
  \MEv(\vec{u})|\leq t+2|D_M|$ and there are at most
  $3^{2|D_M|+t}$ possible assignments to these coordinates. Since
  $\vec{v}$ and $\vec{u}$ are equal on all other coordinates, \ie,
  the coordinates in $[d]\setminus (\HSET(\vec{v},\vec{u})\cup
  \MEv(\vec{v})\cup \MEv(\vec{u}))$,
  we obtain that any such vector $\vec{m}$ disagrees with $\vec{v}$ (or
  $\vec{u}$)
  on at most $t/2$ of these coordinates; since otherwise
  $|\HSET(\vec{v},\vec{m})\cap\HSET(\vec{u},\vec{m})|>t/2$ and we
  would have removed $\vec{m}$ from
  $M\setminus(D_M\cup\{\vec{v},\vec{u}\})$. Since the vectors in $S$
  agree on all coordinates in $\HSET(\vec{v},\vec{u})\cup \MEv(\vec{v})\cup
  \MEv(\vec{u})$, we can complete all of them in the same manner on
  these coordinates. Moreover, on all other coordinates, $\vec{v}$ is
  equal to $\vec{u}$ and we can complete all vectors in $S$ by setting
  them equal to $\vec{v}$ (or equivalently $\vec{u}$) on all other coordinates.
  Let $S'$ be the set of vectors obtained from $S$ after completing
  the vectors in this manner. Then $\HSET(\vec{m},\vec{m}') \subseteq [d]\setminus (\HSET(\vec{v},\vec{u})\cup
  \MEv(\vec{v})\cup \MEv(\vec{u}))$; moreover,
  $|\HSET(\vec{m},\vec{m}')|\leq|(\HSET(\vec{v},\vec{m})\cap\HSET(\vec{u},\vec{m}))\cup(\HSET(\vec{v},\vec{m}')\cap\HSET(\vec{u},\vec{m}'))|\leq
  2t/2=t$ for every two vectors $\vec{m}$ and $\vec{m}'$ in
  $S'$. Therefore, the vectors in $S'$ form a \DCLUS{} with diameter
  at most $t$. Moreover, the vectors in $S'$ also form a \RCLUS{} with
  radius at most $t/2$, as witnessed by a center equal to
  $\vec{v}$ (or equivalently $\vec{u}$) on all coordinates in
  $[d]\setminus
  (\HSET(\vec{v},\vec{u})\cup\MEv(\vec{v})\cup\MEv(\vec{u}))$, and
  equal to any vector in $S'$ on all other coordinates.
\end{proof} \fi
We will now introduce notions required to reduce the number of coordinates.
Let $|\DCOOR(M)|$, for $M \subseteq \{0,1,\blank\}^d$, be the set of
all coordinates $i$ such that at least two vectors in $M$ disagree on
their $i$-th coordinate, \ie, there are two vectors $\vec{y},\vec{y}'\in M$
such that $\{\vec{y}[i],\vec{y}'[i]\}=\{0,1\}$.
Intuitively,
$\DCOOR(M)$ is the set of \emph{important coordinates}, since all other
coordinates can be safely removed from the instance; this is because
they can always be completed to the same value.
We are now ready to show that \DIAMCq{} has a Turing kernel
parameterized by $k+r+\MEP(M)$.
\fi
\iflong \begin{THE} \fi \ifshort \begin{THE}[$\spadesuit$] \fi\label{thm:DIAM_k_r_comb}
  \DIAMCq{} parameterized by $k+r+\MEP(M)$ has a Turing-kernel
  containing at most $n=k3^{2\MEP(M)+r}+\MEP(M)+2$ vectors, each
  having at most $\max\{r(n-1)+\MEP(M),\binom{\MEP(M)}{2}(r+1)\}$ coordinates.
\end{THE}
\iflong \begin{proof}
    We distinguish three cases: (1) the solution cluster contains at
  least two vectors in $M\setminus D_M$, (2) the solution cluster
  contains exactly one vector in $M\setminus D_M$, and (3) the
  solution cluster does not contain any vector in $M\setminus D_M$.

  We start by showing the result for case (1).
  Our first goal is to reduce the number of vectors in $M\setminus
  D_M$. As a first step, we
  guess two vectors $\vec{v}$ and $\vec{u}$ in $M\setminus D_M$ that
  are farthest apart in the solution cluster w.r.t.\ to all vectors in
  $M\setminus D_M$; this is possible since we assume that the cluster
  contains at least two vectors in $M\setminus D_M$.
              We then guess the exact distance, $t$, between $\vec{v}$
  and $\vec{u}$ in a completion leading to the cluster. Note that $t$ can
  be any value between $\HDIST(\vec{v},\vec{u})$ and
  $\HDIST(\vec{v},\vec{u})+2|D_M|$, but is at
  most $r$.
  We now remove all vectors $\vec{m}$
  from $M$ for which either $\HDIST(\vec{v},\vec{m})>t$,
  $\HDIST(\vec{u},\vec{m})>t$, or
  $|\HSET(\vec{v},\vec{m})\cap\HSET(\vec{u},\vec{m})|>t/2$.
  Note that this is safe by of Lemma~\ref{lem:krcom-prune}.
  It now follows from Lemma~\ref{lem:krcom-kernel} that if $M\setminus
  (D_M\cup\{\vec{v},\vec{u}\})$
  contains more than $k3^{2|D_M|+t}$ vectors, then we can return a
  trivial \YES\hy instance of \DIAMCq{}.
  Otherwise, we obtain that $|M\setminus
  D_M|\leq k3^{2|D_M|+t}+2$. We now add back the vectors
  in $D_M$. Clearly, if a vector in $D_M$ differs in more than $t$
  coordinates from $\vec{v}$, then it
  cannot be part of a \DCLUS{} containing $\vec{v}$, and we can safely remove it from
  $D_M$. Hence every vector in $M$ now differs from $\vec{v}$ in at most $t$
  coordinates, which implies that $\DCOOR(M)\leq
  t(|M|-1)+|D_M|$. We now remove all coordinates outside of
  $\DCOOR(M)$ from $M$, since they can always be completed in the
  same manner for all the vectors. Now the remaining instance is a
  kernel containing at most $m=k3^{2|D_M|+t}+2+|D_M|$ vectors, each
  having at most $t(m-1)+|D_M|$ coordinates. This completes the proof for
  the case where the solution cluster contains at least two vectors in
  $M\setminus D_M$.

  In the following let $t=r$.
  For the second case, \ie, the case that the solution cluster
  contains exactly one vector in $M\setminus D_M$, we first guess the
  vector say $\vec{m}$ in $M\setminus D_M$ that will be included in the solution
  cluster and remove all other vectors from $M \setminus
  D_M$.  This leaves us with at most $|D_M|+1$ vectors, and it only remains
  to reduce the number of relevant coordinates. Since we guessed
  that $\vec{m}$ will be in the solution cluster, we can now safely
  remove all vectors $\vec{m}' \in M$, with
  $\HDIST(\vec{m},\vec{m}')>t$. Every remaining vector
  differs from $\vec{m}$ in at most $t$ coordinates, and hence
  $\DCOOR(M) \leq t|D_M|+|D_M|$, which gives us the desired kernel.

  For the third case, \ie, the case that the solution cluster
  contains no vectors in $M\setminus D_M$, we first remove all vectors
  in $M\setminus D_M$. This leaves us with $|D_M|$ vectors, and it
  only remains to reduce the number of coordinates. To achieve this,
      we compute a set of coordinates
  that preserves the distance up to $t$ between any pair of vectors.
  Namely, we compute a set $D$ of relevant coordinates, starting from
  $D=\emptyset$, by adding the following coordinates to $D$
  for every two distinct vectors $\vec{m}$ and $\vec{m}'$ in
  $M$:
  \begin{itemize}
  \item if $|\HSET(\vec{m},\vec{m}')|\leq t$, we add
    $\HSET(\vec{m},\vec{m}')$ to $D$ and otherwise
  \item we add an arbitrary subset of exactly $t+1$ coordinates in
    $\HSET(\vec{m},\vec{m}')$ to $D$.
  \end{itemize}
  Let $M_D$ be the matrix obtained from $M$ after removing all
  coordinates/columns not in $D$. We claim that $(M,k,r)$ and $(M_D,k,r)$ are
  equivalent instances of \DIAMCq{}. This follows because the vectors
  in any \DCLUS{} for $(M_D,k,r)$ can be completed in the same manner
  on all coordinates not in $D$. Since $|D|\leq
  \binom{|D_M|}{2}(t+1)$, the remaining instance has at most $|D_M|$
  vectors each having at most $\binom{|D_M|}{2}(t+1)$ coordinates.
\end{proof} \fi

\newcommand{\core}{center}

\newcommand{\allowedSet}{\ensuremath{\Lambda}}

\subsection{\DIAMCq{} Parameterized by $r+\MEP(M)$}
With Theorem~\ref{thm:DIAM_k_r_comb} in hand, we can move on to establishing the fixed-parameter tractability of \DIAMCq\ parameterized by $r+\MEP(M)$.
At the heart of our approach lies a new technique for analyzing the structure of vectors through sunflowers in their set representations, which we dub \emph{iterative sunflower harvesting}. 
We first preprocess the instance to establish some basic properties. We then show a general result about sunflowers that allows us to derive a succinct representation of the solution cluster---in particular, this guarantees that the hypothetical solution can be described by a bounded number of sunflower cores. Finally, we proceed to ``harvesting'' these sunflowers cores using a branching procedure, thus computing the solution cluster.

\subsubsection{Preprocessing the Instance}
Let $(M,k,r)$ be an instance of \DIAMCq{}, let $M^*$ be a completion of $M$ that contains a maximum size \DCLUS{}, and fix $P^*$ to be such a maximum size \DCLUS{} in $M^*$. We also fix $D_M$ to be a $\MEP(M)$\hy deletion set. The goal of the algorithm is to find $M^*$ and $P^*$.
 
If $k< \MEP(M)+2$, then we can use the algorithm in Theorem~\ref{thm:DIAM_k_r_comb} to obtain a Turing kernel whose size is a function of $r+\MEP(M)$, which, in turn, would imply that \DIAMCq{} if \FPT{} parameterized by $r+\MEP(M)$, thus giving the desired result. We assume henceforth that $P^*$ contains the completion of at least two vectors in $M\setminus D_M$. We can guess the subset $P_R$ of $D_M$ that will be completed to vectors in $P^*$, and restrict our attention to finding a
\DCLUS{} in $M\setminus D_M$ of size $|P^*\setminus D_M|$. We will do so by enumerating all such \DCLUS{}'s in $M\setminus D_M$ and at the end check whether one of them, together with $P_R$, can be extended into a \DCLUS{}.

Therefore, we will focus on what follows on finding a cluster in $M\setminus D_M$. We first
guess two vectors $\vec{v}$ and $\vec{u}$ in $M\setminus D_M$, together with their completions $\vec{v}^*$ and $\vec{u}^*$, respectively, such that $\vec{v}^*$ and $\vec{u}^*$ are both in $P^*$ and are the farthest vectors apart in $P^*\setminus D_M$; fix $r_{\max}=\HDIST(\vec{v}^*,\vec{u}^*)$. We remove all other vectors that can be completed into $\vec{v}^*$ or $\vec{u}^*$, and reduce $k$ accordingly; hence, we do not keep duplicates of the two vectors that we already know to be in $P^*$.
We then normalize all vectors in $M$ so that
$\vec{v}^*$ becomes the all-zero vector, \ie, we replace $\vec{v}^*$ by
the all-zero vector, and for every other vector $\vec{w}\neq \vec{v}$,
we replace it with the vector $\vec{w}'$ such that $\vec{w}'[i]=0$ if
$\vec{v}^*[i]=\vec{w}[i]$, $\vec{w}'[i]=\blank$ if $\vec{w}[i]=\blank$, and $\vec{w}'[i]=1$, otherwise. Finally, for each vector $\vec{w}\in M\setminus D_M$, we compute the set $\allowedSet(\vec{w})$ of all completions of $\vec{w}$ at distance at most $r_{\max}$ from both $\vec{v}^*$ and $\vec{u}^*$. Note that $\allowedSet(\vec{w})$ can be computed in $\bigoh(2^{\MEP(M)}\cdot d)$ time for each vector in $M\setminus D_M$, where $d$ is the dimension of the vectors in $M$. We then remove all vectors $\vec{w}$ with $\allowedSet(\vec{w})=\emptyset$ from $M\setminus D_M$.

We will extend the notation $\allowedSet(\vec{w})$ to $\allowedSet_C(\vec{w})$, for a multiset $C$ of vectors in $\{0,1\}^d$, such that $\allowedSet_C(\vec{w})$ is the set of all completions of a vector $\vec{w}$ at distance at most $r_{\max}$ to all vectors in $C$. We are now ready to show that after normalizing the vectors in $P^*$, the multiset $P^*\setminus D_M$ satisfies certain structural properties that we refer to as an \threshold{$r$} subset (of $M^*$); these structural properties allow for a succinct representation of $P^*$.

\subsubsection{Sunflower Fields or Representing Sets by Cores of Sunflowers}

In this subsection, we provide the central component for our algorithm based on
the iterative sunflower harvesting technique. Crucial to this component
is a general structural lemma that allows us to represent a family
of sets in a succinct manner in terms of sunflower cores, which we believe
to be interesting in its own right. We first state the result in its
most general form (for sets), and then show how to adapt it to
our setting.

\begin{definition}
Let $U$ be a universe, $\BBB$ a family of subsets of $U$ and
$\AAA\subseteq \BBB$. We say that $\AAA$ is an \emph{\threshold{$r$}} subfamily of
$\BBB$ (for $r \in \mathbb{N}$) if the following holds for every $t\in\Nat$ and every sunflower
$S\subseteq \AAA$ containing at least $r+1$ sets of cardinality $t$ with
core $C$: $\AAA$ contains every set $B\in \BBB$ of cardinality $t$ such that
$C\subseteq B$.
\end{definition}

Intuitively, this property states that $\AAA$ contains
all sets in $\BBB$ which are super-sets of cores of every sufficiently-large sunflower in $\AAA$ (with sets of the same cardinality).

The connection of this set property to clusters is as follows. We will show that every maximal cluster is an \threshold{$r$} subset of $M^*$. 
Since $P^*$ contains the all-zero vector $\vec{v}^*$, any vector in $P^*$ contains at most $r$ ones. Fix $t \in [r]$, and consider the set of all vectors in $P^*$ containing exactly $t$ ones. The above notion will allow to draw the following assumption: if the aforementioned set of vectors is large, then it must contain a large sunflower and all the vectors in $M$ whose completions share the core of this sunflower \emph{must be} in $P$.  This property will subsequently allow us to represent every hypothetical solution using a bounded number of sunflowers, as we show next.

The crucial insight now is that every \threshold{$r$} subfamily
containing only sets of bounded size admits a succinct representation, where we can
completely describe the set via a bounded number of
sunflower cores. This is made precise in the following lemma.
\iflong
\begin{LEM}
\fi \ifshort \begin{LEM}[$\spadesuit$] \fi
\label{lem:setsofsets}
	Let $\AAA$ and $\BBB$ be two families of sets of cardinality at most
	$r'$ over universe $U$ such that $\AAA \subseteq \BBB$.
	If $\AAA$ is an $r$-saturated subset of $\BBB$,
	then there is a set $\SSS$ of at most $(rr')^{r'}$ subsets of $U$
	such that $\AAA$ is equal to the
	set of all sets $B$ in $\BBB$ satisfying that $S \subseteq B$ for some $S
	\in \SSS$. Moreover, for each set $S \in \SSS$, $S$ is either the core of a
	sunflower in $\AAA$ with at least $r+1$ petals, or $|S|=r'$.
\end{LEM}
\iflong
\begin{proof}
	We show the claim by giving a bounded-depth search tree algorithm
	that computes a rooted (search) tree $T$, where each node $t$ of $T$ is
	associated with a subset $X_t$ of $U$, and in which each leaf node is
	marked as either successful or unsuccessful. Note that the algorithm is not necessarily efficient, and all
	the steps can be brute-forced. The set $\SSS$ then
	contains all sets associated with a successful leaf node of $T$.
	Initially, $T$ only consists of the root node $r$ for which we set
	$X_r:=\emptyset$. The algorithm then proceeds as follows as long as
	$T$ contains a leaf that is not yet marked as either successful or unsuccessful.
	Let $t$ be a leaf node of $T$ (that has
	not yet been marked). We
	first check whether $|X_t|=r'$ and $X_t \in \AAA$, or whether $\AAA$
	contains a sunflower with at least $r+1$ petals and core $X_t$; 		if this is the case, we
	mark $t$ as a successful leaf node. Note that, since $\AAA$ is an
	$r$-saturated subset of $\BBB$, $\AAA$ contains all sets $B \in \BBB$ such
	that $X_t \subseteq B$.
	Otherwise, let $\AAA_X$ be the subset of $\AAA$ containing all sets
	$A \in \AAA$ with $X_t \subseteq A$. If $\AAA_X=\emptyset$, we mark
	$t$ as an unsuccessful leaf node.
	Otherwise, let $\AAA_X'$ be a maximal sunflower in $\AAA_X$ with core $X_t$. Note that
	$\AAA_X' \neq \emptyset$ because $\AAA_X\neq \emptyset$ (and we can choose the
	core of a sunflower containing only one set/petal arbitrarily). Moreover,
	$|\AAA_X'|\leq r$, since otherwise $\AAA$ contains a sunflower with at
	least $r+1$ petals and core $X_t$. We now claim that the set
	$H:=(\bigcup_{A \in \AAA_X'}A)\setminus X_t$ is a hitting set for the sets
	in $\AAA_X$. Assume, for a contradiction, that this is not the case and
	let $A \in \AAA_X$ be such that $H\cap A=\emptyset$. Then
	we can extend $\AAA_X'$ by $A$ contradicting our
	assumption that $\AAA_X'$ is inclusion-wise maximal. Now, for every $h \in
	H$, we add a new child node $t_h$ to $t$ and set $X_{t_h}:=X_t\cup
	\{h\}$. Note that we add at most $|H|\leq r(r'-|X_t|)$ children to $t$.
	
	Let $\SSS$ be the set computed by the above algorithm. Note that $T$
	has height at most $r'$. This is because $|X_t|=l$ for any node $t$
	of $T$ at level $l$ and we stop when there is at most one set in
	$\AAA$ containing $X_t$, which implies that $t$ is marked either as
	successful or unsuccessful leaf node.
	Moreover, since every node of $T$ has at most $rr'$ children,
	we obtain that $T$ has at most $(rr')^{r'}$ (leaf) nodes and
	hence 	$|\SSS|\leq (rr')^{r'}$, as required. 		It remains to show
	that, 	for every set $A$ in $\AAA$, there is a $S\in \SSS$ such that
	$S\subseteq A$.
	Let $A$ be an arbitrary set in $\AAA$. We first
	show that there is a leaf node $l$ such that $X_{l}$ is contained in
	$A$. Towards showing this, first note that $X_r=\emptyset$ is
	clearly contained in $A$ for the root node $r$ of
	$T$. Furthermore, if $t$ is an inner node of $T$ such that $X_t$ is
	contained in $A$, then there is a child $t'$ of $t$ in $T$
	with $X_{t'}$ contained in $A$. This follows from the
	construction of $T$ since $t$ has a child $t''$ for every $h\in H$
	with $X_{t''}:=X_t\cup \{h\}$, where $H$ is a hitting set of
	$\AAA_{X_t}$ and hence also hits $A$. Starting with the root $r$
	of $T$ and going down the tree $T$ by always choosing a child $t$
	such that $X_t$ is contained in $A$, we will eventually end up
	at the leaf node $l$ such that $X_l$ is contained in
	$A$. Finally, because $\AAA_{X_l}$ contains $A$, $l$ is
	marked as a successful leaf node meaning that $X_l\in \SSS$ is contained in $A$.
\end{proof}
\fi 
We will now show how Lemma~\ref{lem:setsofsets} can be employed in our
setting. Let $M,M'\subseteq \{0,1\}^d$ with $M' \subseteq M$. Then $M'$ is an
\emph{\threshold{$r$} subset} of $M$ if $\HSET(M')$ is an
\threshold{$r$} subset of $\HSET(M)$. Herein, one can think of $M$ as
being (a part of) the input matrix and $M'$ as being an inclusion-wise maximal
cluster; we will later show that this property guarantees that $M'$ is an
\threshold{$r$} subset of $M$.
Using this definition, the following
corollary now follows immediately from Lemma~\ref{lem:setsofsets}.
\begin{COR}\label{cor:setsofsets}
	Let $r',r\in \Nat$ and $M,M'\subseteq \{0,1\}^d$ be sets of
	$r'$-vectors such that $M'$ is an $r$-saturated subset of $M$.
	There is a set $\SSS$ of at most $(rr')^{r'}$ subsets of
	$[d]$ such that $M'$ is equal to the set of all
	vectors $\vec{m}$ in $M$ satisfying $S\subseteq \HSET(\vec{m})$
	for some $S\in \SSS$. Moreover, for each set $S\in\SSS$, $S$ is either a core
	of a sunflower in $\HSET(M')$ with at least $r+1$ petals, or $|S|=r'$.
			\end{COR}

We now can show that after normalizing the vectors in $P^*$, the multiset $P^*\setminus D_M$ is an \threshold{$r$} subset of $M^*$.

\ifshort
\begin{LEM}[$\spadesuit$]

\label{lem:r-uniform-PAIR}
	Let $(M,k,r)$ be an instance of \DIAMCq{}, let $M^*$ be a completion of $M$ and let $P^*$ be a \DCLUS{} in $M^*$ of maximum size such that $\vec{0}\in P^*$. Then for every $N\subseteq M^*$, $P^*\setminus N$ is an \threshold{$r$} subset of $M^*\setminus N$.
\end{LEM}
\fi

 \iflong
\begin{LEM}

\ifshort
\begin{LEM}[$\spadesuit$]
\fi
\label{lem:r-uniform-PAIR}
	Let $(M,k,r)$ be an instance of \DIAMCq{}, let $M^*$ be a completion of $M$ and let $P^*$ be a \DCLUS{} in $M^*$ of maximum size such that $\vec{0}\in P^*$. Then for every $N\subseteq M^*$, $P^*\setminus N$ is an \threshold{$r$} subset of $M^*\setminus N$.
\end{LEM}
\fi
\iflong
\begin{proof}
	\iflong
	Let $t\in \Nat$ and $S\subseteq (P^*\setminus N)$ such that
	\begin{itemize}
		\item $|S|>r$,
		\item for all $\vec{w}\in S$ it holds that $|\HSET(\vec{w})|=t$, and
		\item $\FFF=\HSET(S)$ is a sunflower with core $X$.
	\end{itemize}
	\fi
	\ifshort
	Let $t\in \Nat$ and $S\subseteq (P^*\setminus N)$ such that (1) $|S|>r$, (2) for all $\vec{w}\in S$ it holds that $|\HSET(\vec{w})|=t$, and (3) $\FFF=\HSET(S)$ is a sunflower with core $X$.
	\fi
	Let $\vec{m}$ be an arbitrary $t$-vector in $M^*\setminus (N\cup  P^*)$ containing $X$, and let $\vec{c}$ be an arbitrary vector in $P^*$. Then
	$|\HSET(\vec{c})|= \HDIST(\vec{0},\vec{c})\le r$, and it follows from
	Lemma~\ref{lem:CLUS-SF-BASIC-A} that $S$ contains a vector $\vec{n}$
	that has maximum distance to $\vec{c}$ among all $t$-vectors in $M^*$
	containing $X$. Hence, the distance between $\vec{m}$ and $\vec{c}$
	is at most that between $\vec{n}$ and $\vec{c}$,
	which, since both $\vec{n}$ and $\vec{c}$ are in $P^*$, is at most
	$r$. Consequently, the distance between $\vec{m}$ and any vector in
	$P^*$ is at most $r$, contradicting the maximality~of~$P^*$.
\end{proof}
\fi

Since $P^*\setminus N$ is an \threshold{$r$}
subset of $M^*\setminus N$, by Corollary~\ref{cor:setsofsets}, applied separately for each  $r'\in [r]$, there exists a set $\SSS=\{(S_1,r_1),\ldots, (S_\ell,r_\ell) \}$, with $\ell\le \sum_{r'\in [r]}(rr')^{r'}\le r^{2r+1}$, such that $P^*\setminus N$ contains precisely all the vectors $\vec{w}$ in $M^*$, such that for some $(S_i,r_i)$, $i\in [\ell]$, $S_i\in \HSET(\vec{w})$ and $|\HSET(\vec{w})|=r_i$.

We call the pair $(S_i,r_i)$ an
\emph{$r_i$-\core{}} (of $P^*\setminus N$ in $M^*\setminus N$).
We say
that a vector $\vec{w}\in \{0,1,\blank \}^d$ is \emph{compatible} with
$r_i$-\core{} $(S_i,r_i)$ if there is a completion $\vec{w}^*\in
\{0,1\}^d$  of $\vec{w}$, called \emph{witness of compatibility}, such that
$S_i\subseteq \HSET(\vec{w}^*)$ and $|\HSET(\vec{w}^*)|=r_i$. We say that $\vec{w}\in \{0,1,\blank\}^d$
is \emph{compatible} with $\SSS$ if it is compatible with some
$(S_i,r_i)\in \SSS$.

The \emph{size} of an $r_i$-\core{} is the
number of vectors that are compatible with it in $M$. Moreover, for a set
$\SSS=\{(S_1,r_1),\ldots,(S_\ell,r_\ell) \}$ and a multiset $C$ of
vectors from $\{0,1\}^d$, we say that $\SSS$ \emph{defines} $C$, if
every vector $\vec{c}\in C$ is compatible with $\SSS$ and for every
$(S_i,r_i)\in \SSS$ there is a vector in $C$ compatible with
$(S_i,r_i)$. We say that $\SSS$ \emph{properly} defines $C$, if $|\SSS|\le
r^{2r+1}$, $\SSS$ defines $C$, and for every $(S_i,r_i)\in \SSS$
either:
\begin{itemize}
\item $|S_i|=r_i$ and the unique vector that is compatible with
  $(S_i,r_i)$ is in $C$; or
\item $|S_i|<r_i$ and $C$ contains a set $N$ of $r+1$ $r_i$-vectors such that $\HSET(N)$ forms a sunflower with core $S_i$.
\end{itemize}
Note that if every vector in $C$ has at most $r_{\max}$ $1$'s, then since $C$ is an \threshold{$r$} subset of $C$, it follows from Corollary~\ref{cor:setsofsets} that there always exists a set $\SSS$ that properly defines $C$.
\begin{OBS}\label{obs:properlyDefines}
	Let $C$ be a multiset of vectors from $\{0,1\}^d$, with $\max_{\vec{c}\in C}|\HSET(\vec{c})|\le r$. Then there exists a set $\SSS$ of at most $r^{2r+1}$ $r_i$-\core{s} that properly defines $C$.
\end{OBS}

Suppose that we have a correct guess for $\SSS$, then we can already solve the
problem as follows. For each $(S_i,r_i)\in \SSS$ such that
$|S_i|=r_i$, there is only one possible $r_i$-vector that contains
$S_i$. If our guess is correct, then $P^*$ contains at least one vector that can be completed to this particular $r_i$-vector, and by maximality of $P^*$, $P^*$ has to contain all such vectors.
If $(S_i,r_i)\in \SSS$ such that $|S_i|<r_i$, then $P^*$
contains a sunflower containing $r_i$-vectors of size at least $r+1$ whose core is $S_i$. Clearly,
$P^*$ contains all the vectors that can be completed to an
$r_i$-vector containing $S_i$, since Lemma~\ref{lem:r-uniform-PAIR}
holds for any completion $M^*$ of $M$ that contains $P^*$ as a
subset. The following lemma shows that all such vectors can
be completed arbitrarily since all that matters is that their
completion is compatible with $\SSS$.

\iflong
\begin{LEM} 
\fi
\ifshort
\begin{LEM}[$\spadesuit$]
\fi
\label{lem:largeSunflowerArbitrary}
	Let $(M,k,r)$ be an instance of \DIAMCq{}, $M^*$ a completion of $M$, $P^*$ a \DCLUS{} in $M^*$ of maximum size with $\vec{0}\in P^*$, and $N\subseteq M^*$. 	If $\SSS$ properly defines $P^*\setminus N$, then for every pair of vectors $\vec{w}_1, \vec{w}_2\in \{0,1\}^d$ compatible with $\SSS$ it holds that $\HDIST(\vec{w}_1, \vec{w}_2)\le r$.
\end{LEM}
\iflong
\begin{proof}
	Let $\vec{w}_i$, $i\in \{1,2\}$, be compatible with the $r_i$-\core{} $(S_i,r_i)$. If $|S_i|< r_i$, $i\in \{1,2\}$, then let $N_i\subseteq P^*\setminus N$ be the set of $r_i$-vectors such that $\FFF_i=\HSET(N_i)$ is a sunflower with core $S_i$. Note that such a set $N_i$ always exists by the definition of $\SSS$. Moreover, if $|S_i|=r_i$, then $\vec{w}_i$ is uniquely defined and already in $P^*\setminus N$. Hence, if both $|S_1|=r_1$ and $|S_2|=r_2$, then the lemma holds.
	
	If $|S_1|=r_1$ and $|S_2|<r_2$, then $P^*\setminus N$ contains a vector identical to $\vec{w}_1$, $\HSET(\vec{w}_1)=r_1\le r$ and by Lemma~\ref{lem:CLUS-SF-BASIC-A} there is a vector $\vec{v}_2$ in $N_2$ that has the maximum distance to $\vec{w}_1$ among all $r_2$-vectors that contain $C_2$ and $\HDIST(\vec{w}_1,\vec{w}_2)\le r$. The case where $|S_1|<r_1$ and $|S_2|=r_2$ is symmetric.
	
	Finally, if both $|S_1|<r_1$ and $|S_2|<r_2$, then for every $\vec{v}_1\in N_1$ is $\HSET(\vec{v}_1)=r_1\le r$ and by Lemma~\ref{lem:CLUS-SF-BASIC-A} there is a vector $\vec{v}_2$ in $N_2$ that has the maximum distance to $\vec{v}_1$ among all $r_2$-vectors that contain $C_2$. Hence, $\vec{w}_2$ has distance at most $r$ to all vectors in $N_1$. Now $\HSET(\vec{w}_2)=r_2\le r$ and by Lemma~\ref{lem:CLUS-SF-BASIC-A} there is a vector $\vec{v}_1$ in $N_1$ that has the maximum distance to $\vec{w}_2$ among all $r_1$-vectors that contain $C_1$. Since $\HDIST(\vec{w}_2,\vec{v}_1)\le r$, it follows that $\HDIST(\vec{w}_1,\vec{w}_2)\le r$.
\end{proof}
\fi

It is easy to see that there are at most $d^{r+1}$ choices for a pair $(S_i,r_i)$ (i.e., a $r_i$-\core), and hence we have at most $(d^{r+1})^{r^{2r+1}}=d^{\bigoh(r^{2r+2})}$ possible choices for the set $\SSS$ that properly defines $P^*$.
This bound already implies an
\XP-algorithm. To obtain a fixed-parameter algorithm, it suffices to find the correct guess for $\SSS$ in \FPT-time, which is our next goal.

\subsubsection{Iterative Sunflower Harvesting}

We are now ready to describe the iterative sunflower harvesting procedure, which allows us to obtain
the desired \FPT{}-algorithm.
Namely, we show that instead of enumerating all $d^{r^{2r+2}}$
possible sets $\SSS$ of $r_i$-\core{s} to find the one that properly
defines $P^*\setminus D_M$, it suffices to enumerate only
$f(r,\MEP(M))$-many ``important'' $r_i$-\core{s} for each choice of
$\vec{v}^*$ and $\vec{u}^*$ (recall that $\vec{v}^*$ and $\vec{u}^*$
are the two fixed vectors in $P^*\setminus D_M$  that were guessed),
where $f$ is some function that depends only on $r$ and
$\MEP(M)$. Moreover, we can enumerate these possibilities in
\FPT-time.

\newcommand{\someCompletion}[1]{\ensuremath{\zeta}^{#1}}
\ifshort

We compute $\SSS$ by iteratively adding $r_i$-\core{}s
one by one. The main idea is to show that, for any partial solution
$\SSS'$, there is a bounded number of choices for the next
$r_i$-\core{} to add. As a first step in this direction, the following
lemma shows that for $\SSS'$, there is
always a ``large'' $r_i$-\core{} $(S_i,r_i)$ that can be added to
$\SSS'$, i.e., of size at least
a $(2^rr^{2r+1})$-fraction of the remaining vectors. Before we state
the lemma, we introduce the following notations.
If $\vec{w}$ is compatible with $\SSS$, we
will denote by $\someCompletion{\SSS}(\vec{w})$ the set of witnesses
of compatibility for $\vec{w}$ and $\SSS$. 
Recall that,
for a vector $\vec{w}\in \{0,1,\blank \}^d$ and multiset $C$ of
vectors from $\{0,1\}^d$, $\allowedSet_C(\vec{w})$ denotes the set of
all completions of vector $\vec{w}$ at distance at most $r_{\max}$ to
all vectors in $C$, \ie, $\max_{\vec{c}\in
  C}\{\HDIST(\vec{c},\vec{c}_w)\}\le r_{\max}$.

\begin{LEM}[$\spadesuit$]
  \label{lem:largeCore}
Let $P^*$ be a maximum \DCLUS{} in $(M,k,r)$, $\SSS$ the set
  of $r_i$-\core{s} that properly define $P^*\setminus D_M$, and
  $\SSS'\subseteq \SSS$. Moreover, let $C'$ be the multiset of vectors
  $\vec{w}$ in $M\setminus D_M$ with
  $\someCompletion{\SSS'}(\vec{w})\neq \emptyset$ and $C$ the
  multiset containing a vector
  $\vec{w}_c\in\someCompletion{\SSS'}(\vec{w})$ for every
  $\vec{w}\in C'$.
  Finally, let $M'$ be the multiset consisting of all
  the vectors $\vec{w}\in M\setminus (C\cup D_M)$ with
  $\allowedSet_C(\vec{w})\neq \emptyset$.
  Then there exists
  $(S_i,r_i)\subseteq \SSS\setminus\SSS'$ such that at least
  $(|M'|/2^{r_{\max}}-|D_M|)/r^{2r+1}$ vectors in $M'$ are
  compatible with $(S_i,r_i)$.   \end{LEM}

Note that each normalised vector $\vec{w}$ can be compatible with at most $2^{r+\MEP(M)}$
$r_i$-\core{}s $(S_i,r_i)$, since $S_i \subseteq \HSET(\vec{w}^*)$ for
some completion $\vec{w}^*$ of $\vec{w}$. Now it follows from a counting argument
that the number of large $r_i$-centers is at most
$2^{r+\MEP(M)}(2^rr^{2r+1})=2^{2r+\MEP(M)}r^{2r+1}$ and those can be
enumerated in time $\bigoh(2^{r+\MEP(M)}|M|)$ ($\spadesuit$).
By Observation~\ref{obs:properlyDefines},
$|\SSS|$ and hence the depth of the branching algorithm, is at most
$r^{2r+1}$, which implies the following theorem.
\begin{THE}[$\spadesuit$] 
\label{thm:diamfpt}
	\DIAMCq{} is fixed-parameter tractable parameterized by $r+\MEP(M)$.
\end{THE}

\fi

\iflong
The idea is to enumerate possible sets $\SSS$ such that we
compute pairs $(S_i,r_i)$ in $\SSS$ ordered by size from largest to
smallest, give a lower bound on the size of the next largest
$r_i$-\core{}, and show that only $g(r,\MEP(M))$ such large
$r_i$-\core{} exist, for some function $g$. Note that if we guess an
$r_i$-\core{} $(S_i,r_i)$ with $|S_i|<r_i$, then if our guess is
correct, we can, by Lemma~\ref{lem:largeSunflowerArbitrary}, already
complete the vectors that can be completed to $r_i$-vector containing
$S_i$ and include all of them in the \DCLUS{}. The following lemma
shows that if we have some partial \DCLUS{} $C$ computed, then we can
pre-process the remaining instance such that it contains a \DCLUS{}
compatible with $C$ of size at least a $2^r$-fraction of the remaining
instance. This gives us that the intersection of $P^*$ and the
remaining instance has to be large, and in turn, implies that one of
the remaining $r_i$-\core{s} of $P^*$ has size roughly at least
($2^rr^{2r+1}-|D_M|$)-fraction of the remaining instance. Recall that,
for a vector $\vec{w}\in \{0,1,\blank \}^d$ and multiset $C$ of
vectors from $\{0,1\}^d$, $\allowedSet_C(\vec{w})$ denotes the set of
all completions of vector $\vec{w}$ at distance at most $r_{\max}$ to
all vectors in $C$, \ie, $\max_{\vec{c}\in
  C}\{\HDIST(\vec{c},\vec{c}_w)\}\le r_{\max}$.

\iflong
\begin{LEM} 
\fi
\ifshort
\begin{LEM}[$\spadesuit$]
\fi
\label{lem:largeSubClique}
	Let $C$ be a \DCLUS{} of vectors in $\{0,1\}^d$ such that
	$\{\vec{0},\vec{u}\}\subseteq C$ for some $r_{\max}$-vector $\vec{u}\in \{0,1\}^d$, where $r_{\max}\in [r]$. Suppose further that
	$\HDIST(\vec{a},\vec{b})\le r_{\max}$ for every distinct $\vec{a},\vec{b}\in C$, and
	let $M'$ be a multiset of vectors $\vec{w}$ with $\allowedSet_C(\vec{w})\neq \emptyset$. Then there exists $P\subseteq M'$ and a completion $P^*$ of $P$ such that $|P|\ge |M'|/2^{r_{\max}}$ and $P^*\cup C$ is a \DCLUS{} of diameter $r_{\max}$.
\end{LEM}
\iflong
\begin{proof}
	Fix a completion $M^*$ of $M'$ such that each vector $\vec{w}\in M'$ is completed to $\vec{c}_w\in \allowedSet_C(\vec{w})$. Let $S\subseteq \HSET(\vec{u})$, and let $M^*_S$ be the subset of $M^*$ containing precisely all the vectors $\vec{w}$ such that $\HSET(\vec{w})\cap \HSET(\vec{u})=S$. We claim that $C\cup M^*_S$ is a \DCLUS{}. First, note that $\HDIST(\vec{a},\vec{b})\le r_{\max}$ for every distinct $\vec{a},\vec{b}\in C$ and also for every distinct $\vec{a}\in C$ and $\vec{b}\in M^*_S$. It remains to show that $\HDIST(\vec{a},\vec{b})\le r_{\max}$ for every distinct $\vec{a},\vec{b}\in M^*_S$. 	Clearly, $\HDIST(\vec{a},\vec{b})\le |\HSET(\vec{a})|+|\HSET(\vec{b})|-2|S|$. Now $|\HSET(\vec{a})|\le r_{\max}$ and $|\HSET(\vec{b})|\le r_{\max}$ because $\vec{0}\in C$ and hence $|\HSET(\vec{a})|-|S|\le r_{\max}-|S|$ and $|\HSET(\vec{b})|-|S|\le r_{\max}-|S|$. Moreover $\HDIST(\vec{a},\vec{u})= r_{\max}-|S|+|\HSET(\vec{a})|-|S|\le r_{\max}$. Therefore, $|\HSET(\vec{a})|-|S|\le |S|$ and similarly $|\HSET(\vec{b})|-|S|\le |S|$. Hence $\HDIST(\vec{a},\vec{b})\le |\HSET(\vec{a})|+|\HSET(\vec{b})|-2|S|\le \min (2|S|, 2r_{\max}-2|S|)$. It follows that either $|S|\le r_{\max}/2$ and $\HDIST(\vec{a},\vec{b})\le r_{\max}$ or $|S|\ge r_{\max}/2$ and $\HDIST(\vec{a},\vec{b})\le 2r_{\max}-2|S|\le r$. Since there are $2^{r_{\max}}$ possibilities for the set $S$, and each vector in $M^*$ is in one of sets $M^*_S$, the lemma follows.
\end{proof}
\fi

As a corollary, we obtain the following lemma.
Before we state the lemma, we introduce some additional notations that will be useful.
Let $\SSS=\{(S_1,r_1),\ldots,  (S_\ell,r_\ell) \}$ and let $\vec{w}\in \{0,1,\blank\}$. We say that $\vec{w}$ is \emph{compatible} with $\SSS$ if there exists a completion $\vec{w}^*\in \{0,1\}^d$ of $\vec{w}$, called \emph{witness of compatibility}, and a pair $(S_i,r_i)\in \SSS$ such that $S_i\subseteq \HSET(\vec{w}^*)$ and $|\HSET(\vec{w}^*)|=r_i$. If $\vec{w}$ is compatible with $\SSS$, we will denote by $\someCompletion{\SSS}(\vec{w})$ the set of witnesses of compatibility for $\vec{w}$ and $\SSS$. If $\vec{w}$ is compatible with $\{(S',r')\}$, we simple say $\vec{w}$ is compatible with $(S',r')$.
\iflong
\begin{LEM} 
\fi
\ifshort
\begin{LEM}[$\spadesuit$]
\fi
\label{lem:largeCore}
	Let $P^*$ be a maximum \DCLUS{} in $(M,k,r)$, $\SSS$ the set of $r_i$-\core{s} that properly define $P^*\setminus D_M$, $\SSS'\subseteq \SSS$, let $C'$ be the multiset of vectors $\vec{w}$ in $M\setminus D_M$ with $\someCompletion{\SSS'}(\vec{w})\neq \emptyset$ and $C$ the multiset containing a vector $\vec{w}_c\in\someCompletion{\SSS'}(\vec{w})$ for every $\vec{w}\in C'$. Let $M'$ be the multiset consisting of all the vectors $\vec{w}\in M\setminus (C\cup D_M)$ with $\allowedSet_C(\vec{w})\neq \emptyset$. Then there exists $(S_i,r_i)\subseteq \SSS\setminus\SSS'$ such that at least $(|M'|/2^{r_{\max}}-|D_M|)/r^{2r+1}$ vectors in $M'$ are compatible with $(S_i,r_i)$. \end{LEM}

\iflong
\begin{proof}
	By Lemma~\ref{lem:largeSubClique}, $M'$ contains a \DCLUS{} $P'$ of size at least $|M'|/2^{r_{\max}}$ such that $C\cup P'$ can be completed to a \DCLUS{}.
	From Lemma~\ref{lem:largeSunflowerArbitrary}, it follows that
        every vector in $P^*\setminus (D_M\cup C')$\footnote{With a
          slight abuse of notation, we will allow set operations between sets of completed and uncompleted vectors assuming the natural bijection between completed and uncompleted vectors.} is at distance at most $r_{\max}$ to every vector in $C$ and hence it is in $M'$. Hence $C\cup (P^*\setminus (D_M\cup C'))$ is a \DCLUS{} of size $|P^*\setminus D_M|$. Because $C\cup P'$ can be completed to a \DCLUS{}, it follows from maximality of $P^*$ that $|P^*\setminus (D_M\cup C')|\ge |P'|-|D_M|\ge |M'|/2^r-|D_M|$. But every vector in $P^*\setminus (D_M\cup C')$ is compatible with some $(S_i,r_i)\in\SSS\setminus \SSS'$ and
	$|\SSS\setminus\SSS'|\le |\SSS|\le r^{2r+1}$ and the lemma follows.
\end{proof}
\fi

\newcommand{\enumClusters}{\textbf{FindClusters}}

\iflong
The algorithm now enumerates all ``important'' sets of $r_i$-\core{s} by calling procedure \enumClusters$(\{\vec{0}, \vec{u}^*\}, M, \emptyset)$ described in Algorithm~\ref{alg:FindMaxClique}.

\begin{algorithm}[h!]
	\SetAlgoLined
	\KwData{		a multiset $C$ of vectors from $\{0,1\}^d$,
		a multiset $M$ of vectors from $\{0,1,\blank\}^d$ with at most $r_{\max}$ ones,
		and a set $\SSS=\{(S_1,r_1), (S_2,r_2),\ldots, (S_q,r_q)\}$. Moreover, each vector $\vec{c}\in C$ is associated with a distinct vector $\vec{w}_c\in M$ such that $c$ is a completion of $\vec{w}_c$.}
	\KwResult{A set $\CCC = \{(C_1,\SSS_1), \ldots, (C_\ell,\SSS_\ell) \}$}
	\If{$|\SSS|>r^{2r+1}$}{\Return $\emptyset$;}
	$\CCC=\emptyset$;\\
	Let $M'$ be the set of vectors $\vec{w}$ in $M\setminus C$ with
	$\allowedSet_C(\vec{w})\neq\emptyset$;\label{algline:1.5}\\
	\If{$|M'|\le 2^r\cdot|D_M|+1$\label{algline:1.6}}{Add $\{(C,\SSS)\}$ to $\CCC$;\label{algline:1.7}}
	\ForEach{$S$ such that there exists $\vec{w}\in M'$ with $S\subseteq \{i\mid \vec{w}[i]=1\vee \vec{w}[i]=\blank \}$\label{algline:1.9}}{ 	
		\ForEach{$r'\in\{|S|,\ldots, r\}$\label{algline:1.10}}{
			Let $V$ be the multiset of vectors $\vec{a}$ in $M'$ with $\someCompletion{\{(S,r')\}}(\vec{a})\cap \allowedSet_C(\vec{a})\neq \emptyset$;\label{algline:1.11}\\
			Let $V'$ be the multiset that contains for each $\vec{a}\in V$ arbitrary one vector $\vec{a}'\in \someCompletion{(S,r')}(\vec{a})\cap \allowedSet_C(\vec{a})$;\label{algline:1.12}\\
			\If{$|V'|\ge ({|M'|}/2^r-|D_M|)/r^{2r+1}$\label{algline:1.13}}{
				Let $\CCC=\CCC\cup\enumClusters(C\cup V', M, \SSS\cup (S,r'))$;\label{algline:1.14}
			}
	}}
	\Return $\CCC$;

	\caption{The procedure {\enumClusters($C$, $M$, $\SSS$)}.}\label{alg:FindMaxClique}
	
\end{algorithm}
\fi

\iflong
\iflong \begin{LEM} \fi \ifshort \begin{LEM}[$\spadesuit$] \fi\label{lem:runningTimeAlg}
	Let $(M,k,r)$ be an instance of \DIAMCq{}. 	Then the procedure \enumClusters($\{\vec{0},\vec{u}\}$, $M\setminus D_M$, $\emptyset$) runs in \FPT-time parameterized by $r+\MEP(M)$.
\end{LEM}

\iflong \begin{proof}
	We will prove by induction on $r^{2r+1}+1-|\SSS|$ that the algorithm \enumClusters$(C\cup\{\vec{0}, \vec{u}\}, M\setminus D_M, \SSS)$ runs in time $(r^{2r+1}\cdot 2^{4r+\MEP(M)}\cdot |D_M|)^{r^{2r+1}+1-|\SSS|}\cdot |M|^3\cdot d$.

	Clearly, if $r^{2r+1}+1-|\SSS| \le 0$ then the algorithm outputs $\emptyset$ in constant time. Otherwise, the algorithm proceeds to computing the set $M'$. Note that $\vec{w}\in M\setminus D_M$ has at most $\MEP(M)$ $\blank$'s and at most $2^{\MEP(M)}$ completions. For each, we can in time $\bigoh(d\cdot |C|)$ verify if it is at distance at most $r_{\max}$ to all vectors in $C$. The algorithm then continues to go over all pairs $(S,r')$ such that $r'\in [r]$ and there exists $\vec{w}\in M'$ with $S\subseteq \{i\mid \vec{w}[i]=1\vee \vec{w}[i]=\blank \}$. Since $\vec{w}$ has at most $r_{\max}\le r$ ones and at most $\MEP(M)$ $\blank$'s, there are at most $2^{r+\MEP(M)}\cdot |M'|\cdot r$ pairs $(S,r')$. For each, we can compute sets $V$ and $V'$ in time $\bigoh(2^{\MEP(M)}\cdot |M'|\cdot |C|)$.
	If $|V'|< (|M'|/2^r-D_M)/r^{2r+1}$, the algorithm continues to a different choice of $(S,r')$; otherwise, it calls the subroutine \enumClusters$(C\cup\{\vec{0}, \vec{u}\}\cup V', M\setminus D_M, \SSS\cup \{(S,r')\})$. It remains to show that the subroutine \enumClusters$(C\cup\{\vec{0}, \vec{u}\}\cup V', M\setminus D_M, \SSS\cup \{(S,r')\})$ is called in at most $r^{2r+1}\cdot 2^{4r+\MEP(M)}\cdot |D_M|$ branches, because the running time we get from our inductive hypothesis for \enumClusters$(C'\cup\{\vec{0}, \vec{u}\}, M\setminus D_M, \SSS')$, with $|\SSS'|=|\SSS|+1$ dominates the running time so far.
	
	To do so, let us fix a vector $\vec{w}\in M'$ and let us compute in how many sets $V$ it can appear. Vector $\vec{w}$ can appear in $V$ for some pair $(S,r')$ only if $\someCompletion{(S,r')}(\vec{w})\neq \emptyset$. That can only happen if $S\subseteq \{i\mid \vec{w}[i]=1\vee \vec{w}[i]=\blank \}$. As $|\{i\mid \vec{w}[i]=1\vee \vec{w}[i]=\blank \}|\le r_{\max}+\MEP(M)$, it follows that $\vec{w}$ can appear in at most $2^{r+\MEP(M)}$ multisets. Hence, the sum of sizes of multisets $V$ for all pairs $(S,r')$ is at most $2^{r+\MEP(M)}|M'|$ and the lemma follows by distinguishing between two cases depending on whether $|M'|\le 2^{2r}|D_M|$ or $|M'|> 2^{2r}|D_M|$.
\end{proof} \fi

\iflong \begin{LEM} \fi \ifshort \begin{LEM}[$\spadesuit$] \fi\label{lem:correctnessAlg}
	Let $(M,k,r)$ be an instance of \DIAMCq{} and let $P^*$ be a \DCLUS{} of maximum size with $\{\vec{0},\vec{u}\}\subseteq P^*$ for some $r$-vector $\vec{u}\in \{0,1\}^d$. Then \enumClusters($\{\vec{0},\vec{u}\}$, $M\setminus D_M$, $\emptyset$) contains a pair $(C\cup\{\vec{0},\vec{u}\},\SSS)$ such that $P^*\setminus (D_M\cup \{\vec{0},\vec{u}\})$ is properly defined by some superset of $\SSS$ and $C$ is a \DCLUS{} defined by $\SSS$ of size $|P^*\setminus (D_M\cup \{\vec{u},\vec{w}\})|$.
\end{LEM}
\iflong \begin{proof}
	Let $P\subseteq M\setminus (D_M\cup \{\vec{u},\vec{v}\})$ be the multiset of vectors that are completed to vectors in $P^*$, and let $\SSS$ properly define $P^*\setminus (D_M\cup \{\vec{u},\vec{v}\})$, and $\SSS'\subsetneq \SSS$.
	Let $C'$ be the set of vectors $\vec{w}$ in $P$ with $\someCompletion{\SSS'}(\vec{w})\neq \emptyset$.
	We will show the following claim that holds whenever $P\setminus C'\neq \emptyset$.
	
	\begin{CLM}\label{clm:FindMaxDiamClustInduction}
		\enumClusters($C^*\cup\{\vec{u},\vec{v}\}$, $M\setminus D_M$, $\SSS'$) for some multiset $C^*$ containing a vector $\vec{w}_c\in\someCompletion{\SSS'}(\vec{w})$ for every $\vec{w}\in C'$ calls as a subroutine \enumClusters($C^{**}\cup\{\vec{u},\vec{v}\}$, $M\setminus D_M$, $\SSS'\cup (S_i,r_i)$), where $(S_i,r_i)\in \SSS\setminus\SSS'$ and $C^{**}$ contains a vector $\vec{w}_c\in\someCompletion{\SSS'\cup\{(S_i,r_i)\}}(\vec{w})$ for every vector $\vec{w}\in P$ with $\someCompletion{\SSS'\cup\{(S_i,r_i)\}}(\vec{w})\neq \emptyset$.
	\end{CLM}
	
	\iflong \begin{claimproof}[Proof of Claim]
		Let $M'$ be the set of vectors computed on line~\ref{algline:1.5}. Note that every vector in $P'\setminus C'$ is compatible with $\SSS\setminus \SSS'$ and hence by Lemma~\ref{lem:largeSunflowerArbitrary} it is in $M'$. Therefore, we will enumerate over all $r_i$-\core{s} in $\SSS\setminus \SSS'$ on lines~\ref{algline:1.9}~and~\ref{algline:1.10}. Let $V$ and $V'$ be the multisets computed on lines~\ref{algline:1.11}~and~\ref{algline:1.12}, respectively. Then $C'\cup V$ contains all vectors $\vec{w}$ in $P$ with $\someCompletion{\SSS'\cup\{(S_i,r_i)\}}(\vec{w})\neq \emptyset$. Moreover, $C^*\cup V'$ contains a vector $\vec{w}_c\in\someCompletion{\SSS'\cup\{(S_i,r_i)\}}(\vec{w})$ for every vector $\vec{w}\in P$ with $\someCompletion{\SSS'\cup\{(S_i,r_i)\}}(\vec{w})\neq \emptyset$. By Lemma~\ref{lem:largeCore}, for some $(S_i,r_i)\in \SSS\setminus \SSS'$ is $|V'|\ge (|M'|/2^r-|D_M|)/r^{2r+1}$.
	\end{claimproof} \fi
	
	It follows from inductive application of Claim~\ref{clm:FindMaxDiamClustInduction} that {\enumClusters($\{\vec{0},\vec{u}\}$, $M\setminus D_M$, $\emptyset$)} calls as a subroutine {\enumClusters($C\cup\{\vec{0},\vec{u}\}$, $M\setminus D_M$, $\SSS'$)}, where $C$ contains a vector $\vec{w}_c\in\someCompletion{\SSS'}(\vec{w})$ for every vector $\vec{w}\in P$. 		It follows that $|C|=|P|=|P^*\setminus (D_M\cup \{\vec{0},\vec{u}\})|$ and $C$ is by its definition defined by $\SSS'\subseteq \SSS$. Since $P^*\setminus (D_M\cup \{\vec{0},\vec{u}\})$ is properly defined by $\SSS$, it follows from Lemma~\ref{lem:largeSunflowerArbitrary} that $C$ is a \DCLUS{}. It remains to show that {\enumClusters($C\cup\{\vec{0},\vec{u}\}$, $M\setminus D_M$, $\SSS'$)} actually adds the pair $(C\cup\{\vec{0},\vec{u}\}, \SSS')$ to the output $\CCC$ on line~\ref{algline:1.7}.
	
	Let $M'$ be the set of vectors computed on line~\ref{algline:1.5}. If $|M'|\ge 2^r\cdot|D_M|+1$, then by Lemma~\ref{lem:largeSubClique} there is a \DCLUS{} of size $|C\cup\{\vec{0},\vec{u}\}|+|D_M|+1>|P^*|$, which contradicts the choice of $P^*$. Hence the condition on line~\ref{algline:1.6} is satisfied and the pair $(C,\SSS)$ is added to $\CCC$.
\end{proof} \fi

\fi

\ifshort
With Lemma~\ref{lem:largeCore}, we can now enumerate all ``important'' sets of $r_i$-\core{s} in FPT-time ($\spadesuit$); crucially, at least one of the identified clusters is maximum. 
With this, we can put everything together to establish the desired result:
\fi

\iflong \begin{THE} \fi \ifshort \begin{THE}[$\spadesuit$] \fi
\label{thm:diamfpt}
	\DIAMCq{} is fixed-parameter tractable parameterized by $r+\MEP(M)$.
\end{THE}
\iflong \begin{proof}
	Let $(M,k,r)$ be an instance of \DIAMCq{}. If $k\le \MEP(M)+2$, we run the algorithm from Theorem~\ref{thm:DIAM_k_r_comb}. Otherwise, every maximum size \DCLUS{} contains at least two vectors in $M\setminus D_M$. Let us fix some maximum cluster $P$ in $(M,k,r)$ together with its completion $P^*$. We branch over all pairs of vectors in $\vec{v}, \vec{u}\in M\setminus D_M$ as the two vectors completed to two vectors that are farthest apart in $P^*$. We branch over all completions $\vec{u}^*$ and $\vec{v}^*$ of vectors $\vec{u}$ and $\vec{v}$ and we normalize the instance so that $\vec{v}^*=\vec{0}$. In the correct branch, it holds after normalization that $\vec{u}^*, \vec{0}\subseteq P^*$. We fix $r_{\max}=|\HSET(\vec{u}^*)|$.
	Let $\SSS$ be the set that properly defines $P^*\setminus (D_M\cup\{\vec{0},\vec{u}^*\})$.
	By Lemmas~\ref{lem:runningTimeAlg}~and~\ref{lem:correctnessAlg} the procedure {\enumClusters($\{\vec{0},\vec{u}\}$, $M\setminus D_M$, $\emptyset$)} computes in \FPT-time a set $\CCC$ of pairs of form $(C_i,\SSS_i)$ such that there exists $(C_i,\SSS_i)\in \CCC$ such that $|C_i|=|P^*\setminus D_M|$ and some subset $\SSS'\subseteq \SSS$ defines $C_i\setminus\{\vec{0},\vec{u}^*\}$. It follows from Lemma~\ref{lem:largeSunflowerArbitrary} that $C_i\cup (P^*\cap D_M)$ is \DCLUS{}. For each $(C_i,\SSS_i)\in \CCC$, we enumerate all $2^{|D_M|}$ subsets of $D_M$. It remains to show that if $P^*\cap D_M$ is a completion of a multiset $P_R\subseteq D_M$, we can find a completion $P_R^*$ or $P_R$ such that $C_i\cup P_R^*$ is a \DCLUS{} in \FPT-time. Since all sets $S'$ with $(S',r')\in \SSS_i$ have size at most $r$ and $|\SSS_i|\le r^{2r+1}$, it suffices to show the following claim:
	\begin{CLM}
		For every $\vec{w}\in P_R$, if coordinate $j\in [d]$ is not in
		$\HSET(\vec{u})\cup \bigcup_{(S',r')\in \SSS_i}S'$, then we can safely set $\vec{w}[j]=0$.
	\end{CLM}
	\iflong \begin{claimproof}[Proof of Claim]
		Let $\vec{w}^*$ be the completion of $\vec{w}$ in $P^*$ and let $\vec{w}'$ be the completion of $\vec{w}$ where we set all $\blank$'s in $\vec{w}$ at coordinates not in $\HSET(\vec{u})$ or any set $S'$ such that $(S',r')\in \SSS_i$ to zero and set all remaining $\blank$'s as in $\vec{w}^*$.
		
		Let $\vec{c}\in C_i$ be compatible with $(S',r')\in \SSS_i$. If $\vec{c}=\vec{u}$ or $r'=|S'|$, then $\vec{w}^*$ and $\vec{w}'$ are the same on coordinates in $\HSET(\vec{c})$ and $|\HSET(\vec{w})'|\le |\HSET(\vec{w})^*|$. Moreover, $\vec{c}$ is also in $P^*$ (or some copy of $\vec{c}$) as it is the unique vector that is compatible with $(S',r')\in \SSS_i$ and $\SSS_i\subseteq \SSS$. Hence $\HDIST(\vec{c},\vec{w}')\le r$.
		If $|S'|<r'$, then $P^*$ contains a set $N$ of $r'$-vectors of size $r+1$ such that $\FFF=\HSET(N)$ is a sunflower with core $S'$. Hence, by Lemma~\ref{lem:CLUS-SF-BASIC-A} there is a vector $\vec{n}\in N$ with maximum distance to $\vec{w}^*$ among all vectors in $\{0,1\}^d$ that contain $S'$. It is rather straightforward to see that $\HDIST(\vec{c},\vec{w}')\le \HDIST(\vec{n},\vec{w}^*)\le r$.
	\end{claimproof} \fi
	
	It follows from the above claim that it suffices to branch over all completions of $P_R$ that set all $\blank$'s not in $\HSET(\vec{u})\cup \bigcup_{(S',r')\in \SSS_i}S'$ to $0$. There are at most $r^{\bigoh(r^2|P_R|)}$ such completions and the theorem follows.
\end{proof} \fi
\fi

\subsection{\DIAMCq{} Parameterized by $k$}

\iflong
In this subsection, we show that \DIAMCq{} parameterized by $k$ is in \XP{}.
\fi
\ifshort
Here, we use an Integer Linear Programming subroutine to show that \DIAMCq{} parameterized by $k$ is in \XP{}.
\fi
Moreover, we also observe in Theorem~\ref{the:diam-w1-k} that, unless \W{1}=\FPT, this cannot be improved to
an \FPT-algorithm even for complete data.

\begin{THE}   \label{the:diam-xp-k}
  \DIAMCq{} is in \XP{} parameterized by $k$.
\end{THE}

\begin{proof}
  Let $(M,k,r)$ be an instance of \DIAMCq{}. The algorithm works by
  enumerating all potential clusters $C$ of size exactly $k$, and then
  uses a reduction to an ILP instance with $f(k)$ variables to check
  whether $C$ can be completed into a cluster. Since
  there are at most $|M|^k$ many potential clusters of size exactly
  $k$, it only remains to show how to decide whether a given set $C$ of exactly $k$
  vectors in $M$ can be completed into a \DCLUS{}.
  Let $M_C$ be the submatrix of $M$ containing only the vectors in
  $C$. Then $M_C$ has at most $3^k$ distinct columns, and moreover, each
  of those columns can be completed in at most $2^k$ possible ways.
  Let $T$ be the set of all columns occurring in $M_C$ and for a column
  $\vec{t} \in T$, let $F(\vec{t})$ be the set of all possible
  completions of $\vec{t}$,
  and let $\#(\vec{t}\ )$ denote the number of columns in $M_c$ equal to
  $\vec{t}$. For a vector $\vec{f}\in \{0,1\}^k$
  (representing the completion of a column), let $T(\vec{f})$ denote the
  subset of $T$ containing all columns $\vec{t}$ with $\vec{f} \in F(\vec{t})$.
  Moreover, for every $i$ and $j$ with $1\leq i <j \leq k$
  (representing the $i$-th and the $j$-th vectors in $C$), we denote by
  $FD(i,j)$ the set of all vectors
  (completions of columns) $\vec{f} \in \{0,1\}^k$ such that
  $\vec{f}[i]\neq \vec{f}[j]$.

  We are now ready to construct an ILP instance $\III$ with at most
  $3^k2^k$ variables that is feasible if and only if $C$ can be
  completed into a \DCLUS{}. $\III$ has one
  variable $x_{\vec{t},\vec{f}}$ for every $\vec{t} \in T$ and every $\vec{f} \in
  F(\vec{t})$ whose value (in a feasible assignment) represents how many
  columns of type $\vec{t}$ in $M_C$ will be completed to $\vec{f}$.
  Moreover, $\III$ has the following constraints:

  \begin{itemize}
  \item One constraint for every $\vec{t} \in T$ stipulating that every
    column of type $\vec{t}$ in $M_C$ is completed in some
    manner:\ifshort \(\sum\limits_{\vec{f} \in F(\vec{t})}x_{\vec{t},\vec{f}}=\#( \vec{t}\ )\).\fi
    \iflong\[ \sum_{\vec{f} \in F(\vec{t})}x_{\vec{t},\vec{f}}=\#( \vec{t}\ ). \] \fi
  \item For every $i$ and $j$ with $1\leq i <j \leq k$
    (representing the $i$-th and the $j$-th vectors in $C$),
    one constraint stipulating that the Hamming distance between
    the $i$-th and the $j$-th vectors in $C$ does not exceed $r$:
    \ifshort$\sum\limits_{\vec{f} \in FD(i,j)\land \vec{t}\in T(\vec{f})}x_{\vec{t},\vec{f}}\leq r$.\fi
    \iflong\[ \sum_{\vec{f} \in FD(i,j)\and \vec{t}\in T(\vec{f})}x_{\vec{t},\vec{f}}\leq r. \]\fi
  \end{itemize}
  This completes the construction of $\III$ and it is straightforward
  to verify that $\III$ has a feasible assignment if and only if $C$
  can be completed to a \DCLUS{}.  Since $\III$ has
  at most $3^k2^k$ variables, and since it is well known that ILP can be solved
  in \FPT-time w.r.t.\ the number of variables~\cite{Lenstra83},
  $\III$ can be solved in \FPT-time w.r.t.\ $k$.
\end{proof} 

\iflong
Next we show that, unless \FPT = \W{1}, we cannot get an \FPT-algorithm.
We start with a generic construction that is used in several hardness proofs throughout the paper as well.

Let $G$ be a graph, where $V(G)=\{v_1, \ldots, v_n\}$ and $m=|E(G)|$, and let $deg(v_i) \leq n-1$ denote the degree of $v_i$ in $G$. Fix an arbitrary ordering ${\cal O}=(e_1, \ldots, e_m)$ of the edges in $E(G)$.
For each vertex $v_i \in V(G)$, define a vector $\vec{a_i} \in \{0, 1\}^m$ to be the incidence/characteristic vector of $v_i$ w.r.t.\ ${\cal O}$; that is, $\vec{a_i}[j]=1$ if $v_i$ is incident to $e_j$ and $\vec{a_i}[j]=0$ otherwise. Afterwards, expand the set of coordinates of these vectors by adding to each of them $n(n-1)$ ``extra'' coordinates, $n-1$ coordinates for each $v_i$, $i \in [n]$; we refer to the $n-1$ (extra) coordinates of $v_i$ as the ``private'' coordinates of $v_i$. For each $v_i$, $i \in [n]$, we will choose a number $x_i \in \{0, \ldots, n-1\}$, where the choice of the number $x_i$ will be problem dependent, and we will set $x_i$ many coordinates among the private coordinates of $v_i$ to 1, and all other extra private coordinates of $v_i$ to 0. Let $M=\{\vec{a_i} \mid i \in [n]\}$ be the set of expanded vectors, where $\vec{a_i} \in \{0, 1\}^{m+n(n-1)}$ for $i \in [n]$. We have the following straightforward observation:

\begin{OBS}
	\label{obs:genericreduction}
	For each $v_i$, where $i \in [n]$, the number of coordinates in $\vec{a_i}$ that are equal to 1 is exactly $deg(v_i) + x_i$, and
	two distinct vertices $v_i, v_j$ satisfy $\HDIST(\vec{a_i},\vec{a_j}) =deg(v_i)+x_i + deg(v_j) +x_j$ if $v_i$ and $v_j$ are nonadjacent in $G$ and $\HDIST(\vec{a_i},\vec{a_j}) =deg(v_i)+x_i + deg(v_j) +x_j -2$ if $v_i$ and $v_j$ are adjacent.
\end{OBS}

Throughout the paper, we denote by ${\cal R}$ the polynomial-time reduction that takes as input a graph $G$ and returns the set of vectors $M$ described above. We now complete the complexity landscape for \DIAMCq\ with the two missing lower-bounds results.
\fi

\iflong \begin{THE} \fi \ifshort \begin{THE}[$\spadesuit$] \fi\label{the:diam-w1-k}
	\DIAMCq{} is \Weft\emph{[1]}-hard parameterized by $k$ even if $\MEP(M)=0$.
\end{THE}

\iflong \begin{proof}
	We prove \W{1}-hardness by giving a polynomial-time \FPT{} reduction from {\sc Clique}, which is \Weft[1]-hard~\cite{DowneyFellows13}.
	Given an instance $(G, k)$ of {\sc Clique}, where $V(G)=\{v_1, \ldots, v_n\}$, we set $x_i=n-1-deg(v_i)$ for $i \in [n]$, and apply the polynomial-time reduction ${\cal R}$ to $G$ to produce the set of vectors $M$. The reduction from {\sc Clique} to \DIAMCq{} produces the instance $\inst=(M, k, 2n-4)$ of \DIAMCq{} with $\MEP(M)=0$; clearly, this reduction is a polynomial-time \FPT-reduction.
	
	By Observation~\ref{obs:genericreduction}, for any two distinct vertices $v_i, v_j \in V(G)$, $\HDIST(\vec{a_i},\vec{a_j}) =2n-2$ if $v_i$ and $v_j$ are nonadjacent and $\HDIST(\vec{a_i},\vec{a_j}) =2n-4$ if $v_i$ and $v_j$ are adjacent.
	The proof that $(G, k)$ is a \yes-instance of {\sc Clique} iff $(M, k, 2n-4)$ is a \yes-instance of \DIAMCq{} is now straightforward.
	
\end{proof} \fi
We note that our second result
also establishes the \W{1}-hardness of \RADCq\ (since both problems coincide when $r=0$).

\iflong \begin{THE} \fi \ifshort \begin{THE}[$\spadesuit$] \fi
  \label{thm:diamrad-Whard-k}
 \DIAMCq\ and \RADCq\  are \W{1}\hy hard parameterized by $k$ even if $r=0$.
 \end{THE}
\iflong \begin{proof}
              We provide a parameterized reduction from {\sc Independent Set}, which is well known to be \W{1}-complete when parameterized by the size of the sought-after set~\cite{DowneyFellows13}.
        Consider an instance $(G, k)$ of {\sc Independent Set} where $E(G)=\{e_1,\dotsc,e_m\}$, and let us fix an arbitrary ordering $\prec$ on $V(G)$.
  Construct the multiset $M$ of vectors as follows: for each $v\in V(G)$, add a vertex $\vec{v} \in
  \{0,1,\blank\}^{m}$ to $M$ such that
  \begin{itemize}
  \item $\vec{v}[i]=\blank$ if $e_i$ is not incident with $v$,
  \item $\vec{v}[i]=0$ if $e_i=\{u,v\}$ is incident with $v$ and $v\prec u$, and
  \item $\vec{v}[i]=1$ otherwise.
  \end{itemize}

    This completes the construction of $M$; the reduction outputs the instance $(M,k,0)$ of \RADCq\ (or \DIAMCq).

    For correctness, it is straightforward to verify that if $(G,k)$ is a \YES-instance, then $M$ admits a completion $M^*$ and a subset $S\subseteq M^*$ containing at least $k$ identical vectors (indeed, the set of vectors corresponding to an independent set $X$ in $G$ can only contain at most $1$ non-$\blank$ entry on each coordinate). On the other hand, if $M$ admits a completion $M^*$ and a subset $S\subseteq M^*$ containing at $k$ identical vectors, we observe that the vertices corresponding to the vectors in $S$ must form an independent set: indeed, any edge between two such vertices would prevent them to be completed in an identical way.
\end{proof} \fi

\section{Finding a RAD-Cluster in Incomplete Data}

In this section, we present our results for \RADCq{}. We will show that
\RADCq{} is \FPT{} parameterized by
$k+r+\MEP(M)$, and is in \XP{} parameterized by $r+\MEP{}$
alone. Notably, the degree of the polynomial in the run-time of our
\XP{} algorithm grows only logarithmically in $r$ and the algorithm
can be employed to solve the \textsc{Closest String
  with Outliers} problem~\cite{boucherma,BulteauS20}.

Before proceeding to the main contributions of this section, we
observe that, by combining the trivial branching procedure, in which we
branch over all sets of $k$ vectors from $M$ (where in each branch we
proceed under the assumption that all vectors outside of the set can
be deleted), with a previous result of Hermelin and
Rozenberg~\cite[Theorem 2]{hermelin}, which solves the special case of
\RADCq{} for $k=|M|$, we obtain:

\begin{OBS}\label{obs:XPk}
\RADCq{} parameterized by $k$ is in \XP.
\end{OBS}

Together with the previously-established Theorem~\ref{thm:diamrad-Whard-k} (showing the \W{1}\hy hardness for
\RADCq{} w.r.t.\ $k$ even for $r=0$), this gives us an almost
complete picture of the parameterized complexity of \RADCq{} for
any combination of the parameters $k$, $r$, $\MEP(M)$. The only two
questions that remain open are whether the \XP{} result for $r+\MEP(M)$ can
be improved to an \FPT-result (as this has been the case for
\DIAMCq{}), and whether it is possible to obtain an \FPT-algorithm
either for parameter $k$ or $k+\MEP(M)$. As a first step in this direction, we
present an \FPT-approximation scheme for parameter
$r+\MEP(M)$ in Section~\ref{sec:fpt-appr}.
The following observation will be useful:

\begin{OBS}\label{obs:center}
 Given a (complete) vector $\vec{s} \in \{0, 1\}^d$, in time $\bigoh(|M|d)$ we can decide if $\vec{s}$ is the center of a \RADC{} of $k$ vectors in $M$. 
\end{OBS}

Observation~\ref{obs:center} is straightforward since we can find all vectors $\vec{w}\in M$ that can be completed to a vector~$\vec{w}^*$ at distance at most \(r\) from \(\vec{s}\) by letting \(\vec{w}^*[i] = \vec{s}[i]\) wherever \(\vec{w}^*[i] = \blank\), and then decide whether $\vec{w}^*$ is such a vector by computing  \(\HDIST(\vec{w}^*,\vec{s})\).

\subsection{\RADCq{} Parameterized by $k+r+\MEP(M)$.}

We start by showing that, as in the case of \DIAMCq{}, \RADCq{} parameterized by $k+r+\MEP(M)$ has a Turing
kernel. The approach is similar to that
in Subsection~\ref{sssec:diam-k-r-comb}\ifshort.\fi \iflong and will, in particular, make use of Lemma~\ref{lem:krcom-prune}
and Lemma~\ref{lem:krcom-kernel}.\fi
\iflong \begin{THE} \fi \ifshort \begin{THE}[$\spadesuit$] \fi\label{thm:ANY_k_r_comb}
  \RADCq{} parameterized by $k+r+\MEP(M)$ has a Turing-kernel
  containing at most $n=k3^{\MEP(M)+2r}+\MEP(M)+2$ vectors, each
  having at most $\max\{2r(n-1)+\MEP(M),\binom{\MEP(M)}{2}(2r+1)\}$ coordinates.
\end{THE}
\iflong \begin{proof}
    In the following, let $D_M$ be a deletion set for $M$ such that every
  vector in $M\setminus D_M$ has at most $\MEP(M)$ missing entries.
  We distinguish three cases: (1) the solution cluster contains at
  least two vectors in $M\setminus D_M$, (2) the solution cluster
  contains exactly one vector in $M\setminus D_M$, and (3) the
  solution cluster does not contain any vector in $M\setminus D_M$.

  We start by showing the result for case (1).
  Our first goal is to reduce the number of vectors in $M\setminus
  D_M$. As a first step, we
  guess two vectors $\vec{v}$ and $\vec{u}$ of $M\setminus D_M$ that
  are farthest apart in the cluster w.r.t.\ to all vectors in
  $M\setminus D_M$; this is possible since we assume that the cluster
  contains at least two vectors in $M\setminus D_M$.
              We then guess the exact distance, say $t$, between $\vec{v}$
  and $\vec{u}$ in a completion leading to the cluster. Note that $t$ can
  be anywhere between $\HDIST(\vec{v},\vec{u})$ and
  $\min \{\HDIST(\vec{v},\vec{u})+2|D_M|\}$, but also at most $2r$.
  We now remove all vectors $\vec{m}$
  from $M$ for which either $\HDIST(\vec{v},\vec{m})>t$,
  $\HDIST(\vec{u},\vec{m})>t$, or
  $|\HSET(\vec{v},\vec{m})\cap\HSET(\vec{u},\vec{m})|>t/2$.
  Note that this is safe because of Lemma~\ref{lem:krcom-prune}.
  It now follows from Lemma~\ref{lem:krcom-kernel} that if $M\setminus
  (D_M\cup\{\vec{v},\vec{u}\})$
  contains more than $k3^{|D_M|+t}$ vectors, then we can return a
  trivial \YES\hy instance of \RADCq{}.
                                                                            Otherwise, we obtain that $|M\setminus
  D_M|\leq k3^{|D_M|+t}+2$. We now also need to add back the vectors
  in $D_M$. Clearly, if a vector in $D_M$ differs in more than $t$
  coordinates from $\vec{v}$, it
  cannot be part of an \RCLUS{} containing $\vec{v}$, and we can safely remove it from
  $D_M$. Hence every vector in $M$ now differs from $\vec{v}$ in at most $t$
  coordinates, which implies that $\DCOOR(M)\leq
  t(|M|-1)+|D_M|$. We now remove all coordinates outside of
  $\DCOOR(M)$ from $M$, since these can always be completed in the
  same manner for all the vectors. Now the remaining instance is a
  kernel containing at most $m=k3^{|D_M|+t}+2+|D_M|$ vectors each
  having at most $t(m-1)+|D_M|$ coordinates. This completes the proof for
  the case that the solution cluster contains at least two vectors in
  $M\setminus D_M$.

  In the following, let $t=2r$.

  For the second case, \ie, the case that the solution cluster
  contains exactly one vector in $M\setminus D_M$, we first guess the
  vector say $\vec{m}$ in $M\setminus D_M$ that will be included in the solution
  cluster and then remove all other vectors from $M \setminus
  D_M$.  This
  already leaves us with at most $|D_M|+1$ vectors and it only remains
  to reduce the number of relevant coordinates. Because we guessed
  that $\vec{m}$ will be in the solution cluster, we can now safely
  remove all vectors $\vec{m}' \in M$, with
  $\HDIST(\vec{m},\vec{m}')>t$. Now every remaining vector
  differs from $\vec{m}$ in at most $t$ coordinates and hence
  $\DCOOR(M) \leq t(|D_M|+1)+|D_M|$, which gives us the desired kernel.

  For the third case, \ie, the case that the solution cluster
  contains no vectors in $M\setminus D_M$, we first remove all vectors
  in $M\setminus D_M$. This leaves us with only $|D_M|$ vectors and it
  only remains to reduce the number of coordinates. To achieve this,
      we compute a set of coordinates
  that preserves the distance up to $t$ between any pair of vectors.
  Namely, we compute a set $D$ of relevant coordinates starting from
  $D=\emptyset$ by adding the following coordinates to $D$
  for every two distinct vectors $\vec{m}$ and $\vec{m}'$ in
  $M$:
  \begin{itemize}
  \item if $|\HSET(\vec{m},\vec{m}')|\leq t$, we add
    $\HSET(\vec{m},\vec{m}')$ to $D$ and otherwise
  \item we add an arbitrary subset of at most $t+1$ coordinates in
    $\HSET(\vec{m},\vec{m}')$ to $D$.
  \end{itemize}
  Let $M_D$ be the matrix obtained from $D$ after removing all
  coordinates/columns in $D$.
      It is now
  straightforward to show that $(M,k,r)$ and $(M_D,k,r)$ are
  equivalent instances of \RADCq{}. Since $|D|\leq
  \binom{|D_M|}{2}(t+1)$, the remaining instance has at most $|D_M|$
  vectors each having at most $\binom{|D_M|}{2}(t+1)$ coordinates.
\end{proof} \fi

\subsection{\RADCq{} Parameterized by $r+\MEP(M)$}

While, it is relatively easy to see that \RADCq{} parameterized by
$r+\MEP(M)$ can be solved in time $f(\MEP(M),r)n^{\bigoh(r)}$,
here we provide a more efficient algorithm by
reducing the degree of the polynomial in the run-time from $\bigoh(r)$
to $\log r$. Moreover, our algorithm can be
applied to the \textsc{Closest String with Outliers}
problem~\cite{BulteauS20}.

\ifshort \begin{THE}[$\spadesuit$] \label{thm:MAX-ANY-r+comb}
  \RADCq{} can be solved in time $\bigoh(|M|2^{\MEP(M)}(|M|(2^{2r}+d))^{\log r+1})$
  and is therefore in \XP{} parameterized $r+\MEP(M)$.
\end{THE}

\begin{proof}[Proof Sketch]
The main ideas behind the algorithm are captured by the
following definition and discussions.
Let $F \subseteq [d]$ and let $t$ be an integer, where $0 \leq t \leq r$.
We say that a vector $\vec{v}\in \{0,1\}^d$
is an \emph{$(F,t)$-seed} for a center $\vec{c} \in
\{0,1\}^d$ of a solution for $(M,k,r)$ if it satisfies:
\begin{description}
\item[(C1)] $\vec{c}$ agrees with $\vec{v}$ on all coordinates in $F$; and
\item[(C2)] $\vec{c}$ differs from $\vec{v}$ on at most $t$ coordinates
  outside of $F$.
\end{description}
We can show the following statement ($\spadesuit$). If $\vec{v}$ is an $(F,t)$-seed for $\vec{c}$,
then either $\vec{v}$ is the center of a solution for $(M,k,r)$, or
there is a vector $\vec{m} \in M$ with $r < \HDIST(\vec{v},\vec{m})
\leq 2r$ and a subset $C \subset D\setminus F$, where
$D=\HSET(\vec{v},\vec{m})$, such that the vector $\vec{v}'$ obtained
from $\vec{v}$ by complementing all coordinates in $C$ is an $(F\cup
D,t/2)$-seed for $\vec{c}$.  Note that testing the former possibility, that is, whether a vector $\vec{v} \in \{0,1\}^d$ is
  a center of a solution for $(M,k,r)$, can be done in time
  $\bigoh(|M|d)$ by Observation~\ref{obs:center}.

Since there at most $M2^{2r}$ possibilities for $\vec{m}$ and $C$, we
can use the above statement to obtain a $(F\cup D,t/2)$-seed from a given
$(F,t)$-seed. This can be employed within a recursive procedure that,
given a $(\emptyset,r)$-seed for some center $\vec{c}$ of a solution
either obtains a center of a solution or obtains a $(F',0)$-seed
(which itself is the center of a solution), in at
most $\log r$ recursive steps. It only remains to find a
$(\emptyset,r)$-seed for some center $\vec{c}$ of a solution, which
can be achieved by guessing the completion $\vec{v}$ of any vector in
$M\setminus D_M$ that will be in a solution. Note that if $M\setminus D_M$ does not contain a vector in the solution, then $k\le \MEP(M)$ and
 the result follows from  
  Theorem~\ref{thm:ANY_k_r_comb}.
  \end{proof}
\fi

\iflong \begin{THE} \label{thm:MAX-ANY-r+comb}
  \RADCq{} can be solved in time $\bigoh(|M|2^{\MEP(M)}(|M|(2^{2r}+d))^{\log r+1})$
  and is therefore in \XP{} parameterized $r+\MEP(M)$.
\end{THE}
\begin{proof}
  Note first that, by Observation~\ref{obs:center}, testing whether a vector $\vec{c} \in \{0,1\}^d$ is
  a center of a solution cluster for $(M,k,r)$ can be achieved in time
  $\bigoh(|M|d)$. Therefore, it suffices to focus on finding a center for the cluster.
  
  Let $(M,k,r)$ be an instance of $\RADCq$. If $k\le \MEP(M)$, then we can
  use the \FPT-algorithm parameterized by $k+r+\MEP(M)$ implied by
  Theorem~\ref{thm:ANY_k_r_comb}. Otherwise, every \RCLUS{} of size at
  least $k$ contains at least one vector in $M\setminus D_M$.

  Let $F \subseteq [d]$ and let $t$ be an integer, where $0 \leq t \leq r$.
  We say that a vector $\vec{v}\in \{0,1\}^d$
  is an \emph{$(F,t)$-seed} for a center $\vec{c} \in
  \{0,1\}^d$ of a solution for $(M,k,r)$ if it satisfies:

  \begin{description}
  \item[(C1)] $\vec{c}$ agrees with $\vec{v}$ on all coordinates in $F$; and
  \item[(C2)] $\vec{c}$ differs from $\vec{v}$ on at most $t$ coordinates
    outside of $F$.
  \end{description}

  We will show next that if $\vec{v}$ is an $(F,t)$-seed for $\vec{c}$,
  then either $\vec{v}$ is the center of a solution for $(M,k,r)$, or
  there is a vector $\vec{m} \in M$ with $r < \HDIST(\vec{v},\vec{m})
  \leq 2r$ and a subset $C \subset D\setminus F$, where
  $D=\HSET(\vec{v},\vec{m})$, such that the vector $\vec{v}'$ obtained
  from $\vec{v}$ by complementing all coordinates in $C$ is an $(F\cup
  D,t/2)$-seed for $\vec{c}$.

  Towards showing the claim, let $\vec{v}$ be an $(F,t)$-seed for some center $\vec{c}$ of a solution.
  Note that if $\vec{v}$ is not yet
  the
  center of a solution, then the solution with center $\vec{c}$ must contain a vector
  $\vec{m} \in M$ such that $r < \HDIST(\vec{v},\vec{m})$ and $\HDIST(\vec{m},\vec{c})\leq
  r$, and hence, $\HDIST(\vec{v},\vec{m})\leq 2r$. Let $C$ be the subset of coordinates in $D\setminus F$,
  on which $\vec{c}$ and $\vec{v}$ differ, and let $\vec{v}'$ be the
  vector obtained from $\vec{v}$ by complementing the values of the
  coordinates in $C$. Then $\vec{c}$ agrees with $\vec{v}'$ on all
  coordinates in $F\cup D$. Moreover, since $\vec{v}$ differs from
  $\vec{c}$ in at most $t$ coordinates, it follows that $\vec{v}'$
  differs from $\vec{c}$ in at most $t-|C|$ coordinates. In addition,
  $\vec{v}'$ differs from $\vec{c}$ in at most $r-(|D|-|C|)\leq |C|-1$
  coordinates; this is true since $\vec{c}$ differs from $\vec{m}$ in at least $|D|-|C|$ and at
  most $r$ coordinates. Therefore, $\vec{v}'$ differs
  from $\vec{c}$ in at most $\min\{t-|C|,|C|-1\}\leq t/2$ coordinates,
  implying that $\vec{v}'$ is a $(F\cup D,t/2)$-seed for $\vec{c}$.

  \newcommand{\findCenter}{\textbf{findCenter}}

  Now  suppose that we are given a vector $\vec{v} \in \{0,1\}^d$ that
  is an $(F,t)$-seed for the center $\vec{c}$ of a solution for
  $(M,k,r)$, and such that $\vec{v}$ itself is not the center of a solution
  for $(M,k,r)$. Then, the above claim shows us how to compute in time
  $\bigoh(|M|d2^{2r})$ a set $V' \subseteq \{0,1\}^d$ of
  at most $|M|2^{2r}$ vectors, that contains a vector $\vec{v}'$ that is an $(F',
  t/2)$-seed for $\vec{c}$ for some $F'$ satisfying $F \subseteq F'$. This
  procedure can now be used within a recursive procedure
  \findCenter$(\vec{v}, F, t)$ that given
  $\vec{v}$, $F$, and $t$ does the following. It first checks in time
  $\bigoh(|M|d)$ whether $\vec{v}$ is already the center of a
  solution for $(M,k,r)$. If so, it returns $\vec{v}$; otherwise it
  branches on the vectors in $V'$
  and calls itself recursively using the call
  \findCenter$(\vec{v}',F',t/2)$ for every vector $\vec{v}'\in V'$ and
  some $F\subseteq F'$, and returns the center of a solution returned
  by one of those recursive calls, or rejects if all recursive
  calls reject. This way \findCenter$(\vec{v},F,t)$
  is guaranteed to return the center of a solution as long as
  $\vec{v}$ is a $(F,t)$-seed for $\vec{c}$. Moreover, since the recursion depth of \findCenter$(\vec{v},F,t)$ is at
  most $\log t$, it runs in time $\bigoh(|M|(2^{2r}+d)(|M|2^{2r})^{\log t})=\bigoh((|M|(2^{2r}+d))^{\log t+1})$.

  Now suppose that $(M,k,r)$ has a solution with center $\vec{c}$. Since the solution contains at least one vector from $M\setminus
  D_M$, there is a vector $\vec{v}' \in M\setminus D_M$ and a completion
  $\vec{v}$ of $\vec{v}'$ such that $\HDIST(\vec{v},\vec{c})\leq
  r$. The vector $\vec{v}$ is a $(\emptyset,r)$-seed for $\vec{c}$ and
  calling \findCenter$(\vec{v},\emptyset, r)$ will return a center
  for a solution in time $\bigoh((|M|(2^{2r}+d))^{\log t+1})$. Since every
  center $\vec{c}$ of a solution is within the $r$-neighborhood of
  some completed vector in $M\setminus D_M$, we can now decide in time $\bigoh(|M|2^{\MEP(M)}(|M|(2^{2r}+d))^{\log t+1})$ whether
  $(M,k,r)$ has a solution (and if so return a center) by calling
  \findCenter$(\vec{v},\emptyset, r)$ for every completion $\vec{v}$
  of every vector $\vec{v}' \in M\setminus D_M$.
                                                                                                          \end{proof} \fi
We note that the algorithm provided by
the above theorem generalizes a previous algorithm of Marx~\cite[Lemma 3.2]{Marx08} for \textsc{Closest Substring}
to strings that may contain unknown characters. In particular, it lifts the concept of ``generators'' to strings with unknown characters by showing that
there are $\log r$
vectors that can be computed efficiently and that define at most $r\log r$ ``important'' coordinates for the center of
some solution.

\subsection{FPT Approximation Scheme Parameterized by \(r+\MEP(M)\)}\label{sec:fpt-appr}

In this subsection we give an algorithm that, for a given instance \((M,k,r)\) of \RADCq\ and \(\varepsilon\in \mathbb{R}\), where \(0<\varepsilon<1\), computes in FPT-time parameterized by \(r+\MEP(M)+\frac{1}{\varepsilon}\)
a center of a \textsc{Rad}-cluster of size at least \((1-\varepsilon)k\), or it correctly concludes that no \RCLUS{} of size $k$ exists.

\ifshort
\begin{THE} [$\spadesuit$] \label{theshort:radius-approximation}
		Given an instance \((M,k,r)\) of \RADCq\ and \(\varepsilon\in \mathbb{R}\), where \(0<\varepsilon<1\), there exists an \FPT{}
		algorithm \(\mathcal{A}\), parameterized by $r+\MEP(M)+\frac{1}{\varepsilon}$, such that \(\mathcal{A}\) either computes a \RCLUS{} of size at least $(1-\varepsilon)k$, or correctly concludes that \(M\) does not contain a \RCLUS{} of size \(k\).
	\end{THE}
	
\begin{proof}[Proof Sketch]
The algorithm starts by performing a similar branching and pre-processing to the \FPT{} algorithm for \DIAMCq\ parameterized by \(r+\MEP(M)\). Fix \(M^*\) to be a completion of \(M\) that contains a maximum size \RCLUS{}, let \(P^*\) be such a maximum size \RCLUS{} in \(M^*\), and let \(\vec{s}^*\) be a center of \(P^*\). The goal of the algorithm is to find \(\vec{s}^*\), as given \(\vec{s}^*\), by Observation~\ref{obs:center}, we can decide the instance in time $\bigoh(|M|d)$.

If $k<\frac{2\MEP(M)}{\varepsilon}+2$, then we use the algorithm in Theorem~\ref{thm:ANY_k_r_comb} to obtain a Turing kernel\iflong; since the vectors in the Turing kernel have their dimension bounded by a function of the parameter, we can enumerate all possible centers of the \RCLUS{} to solve every sub-instance of the Turing kernel. Now we can assume that $k\ge\frac{2\MEP(M)}{\varepsilon}+2$. We first observe that $|P^*\setminus D_M|$ contains at least two vectors (recall that \(D_M\) denotes the \(\MEP(M)\)-deletion set in \(M\)).
We now\fi\ifshort . Otherwise, we can\fi{} guess two vectors \(\vec{u}\) and \(\vec{v}\) in \(M\setminus D_M\), together with their respective completions \(\vec{u}^*\) and \(\vec{v}^*\), such that \iflong both\fi{} \(\vec{u}^*\) and \(\vec{v}^*\)\iflong{} are in \(P^*\) and\fi{} are the farthest vectors apart in \(P^*\setminus D_M\); fix $r_{\max}=\HDIST(\vec{v}^*,\vec{u}^*)$.\iflong{} Note that \(r_{\max}\le 2r\). Since there are at most \(|M|^2\) pairs of vectors in \(M\setminus D_M\) and each has at most \(\MEP(M)\) missing entries, there are at most \(2^{2\MEP(M)}|M|^2\) possible pairs of  \(\vec{u}^*\) and \(\vec{v}^*\), and our algorithm enumerates all these pairs. What we describe from now on assumes that the guess of  \(\vec{u}^*\) and \(\vec{v}^*\) is correct.\fi{}
We normalize all the vectors in $M$ so that
$\vec{v}^*$ becomes the all-zero vector\iflong , \ie, we replace $\vec{v}^*$ by
the all-zero vector, and for every other vector $\vec{w}\neq \vec{v}^*$,
we replace it with the vector $\vec{w}'$ such that $\vec{w}'[i]=0$ if
$\vec{v}^*[i]=\vec{w}[i]$, $\vec{w}'[i]=\blank$ if $\vec{w}[i]=\blank$, and $\vec{w}'[i]=1$, otherwise\fi .
Finally, for each vector $\vec{w}\in M$, we can in time $r_{\max}\cdot d$, where \(d\) is the dimension of the vectors in \(M\), check if there is a completion $\vec{w}^*$ of \(\vec{w}\) such that the distance from \(\vec{w}^*\) to both \(\vec{v}^*\) and \(\vec{u}^*\) is at most \(r_{\max}\)\iflong{}. This is because we can always set \(\vec{w}^*[i]=0\) for all \(i\in [d]\setminus\HSET(\vec{u}^*)\) such that  \(\vec{w}^*[i]=\blank\), since in that case \(\vec{v}^*[i]=\vec{u}^*[i]=0\) and we are looking for a completion that is close to both \(\vec{v}^*\) and \(\vec{u}^*\). Hence, to decide if there is such completion, we only need to decide how many coordinates of \(\vec{w}^*\) in \(\HSET(\vec{u}^*)\) we need to set to \(0\) and how many we need to set to \(1\). We can now remove all vectors \(\vec{w}\) that cannot be completed to a vector of distance at most $r_{\max}$ from both $\vec{v}^*$ and $\vec{u}^*$.\fi{}\ifshort{}; we remove all vectors \(\vec{w}\) that do not have such a completion.\fi
~Note here that some vectors in \(P^*\cap D_M\)\iflong{} might be further away from $\vec{v}^*$ or $\vec{u}^*$ and they\fi{} could have been removed from \(M\) at this step.
However,\iflong{} since $k\ge\frac{2\MEP(M)}{\varepsilon}+2$, we have\fi{} \(\MEP(M)\le \frac{\varepsilon}{2}k\) and\iflong{}, if our guess of $\vec{v}^*$ and $\vec{u}^*$ is correct,\fi{} we did not remove any vector from \(P^*\setminus D_M\). Hence, after the preprocessing, \(M\) contains a \RCLUS{} with center \(\vec{s}^*\) and at least \(k'=(1-\frac{\varepsilon}{2})k\) vectors. Our goal is to find a center for a \RCLUS{} with at least \((1-\frac{\varepsilon}{2})k' = (1-\frac{\varepsilon}{2})^2k\ge (1-\varepsilon)k\) vectors. For ease of exposition, we let \(\varepsilon' = \frac{\varepsilon}{2}\) and we let \(k'\ge(1-\varepsilon')k\) be the number of vectors of \(P^*\) still in \(M\).
Now we can show the following statement ($\spadesuit$):
After the above pre-processing, in time \(\bigoh(2^{r_{\max}}\cdot |M|)\), we can find a center for a \RCLUS{} of size $\frac{|M|}{2^{r_{\max}}}$.

Therefore, we can assume henceforth that \(k'\ge \frac{|M|}{2^{r_{\max}}}\)\iflong ; otherwise, Lemma~\ref{lem:existence_of_large_rad_cluster} computes a \RCLUS{} of size at least \(k'\)\fi .
Now the algorithm sets \(\vec{s}^*_0 = \vec{v}^* = \vec{0}\) and the goal is to iteratively compute
\(\vec{s}^*_1, \vec{s}^*_2, \vec{s}^*_3, \ldots, \vec{s}^*_{r'}\), \(r'\le r\), such that:

\begin{enumerate}
	\item for all \(i\in [r']\), we have	\(\HSET(\vec{s}^*_i) = \HSET(\vec{s}^*_{i-1})\cup \{c_i\}\) for some coordinate \(c_i\in \HSET(\vec{s}^*)\setminus \HSET(\vec{s}^*_{i-1})\);
	\item for all $j\in[r'-1]$, the number of vectors $\vec{w}$
          with \(\HDIST(\vec{w},\vec{s}^*_{j})\leq r\) is less than \((1-\varepsilon')k'\); and
	\item the number of vectors $\vec{w}$
          with \(\HDIST(\vec{w},\vec{s}^*_{r'})\leq r\) is at least \((1-\varepsilon')k'\).
\end{enumerate}

Let \(\vec{s}^*_i\) be such that \(\HSET(\vec{s}^*_i)\subseteq \HSET(\vec{s}^*)\) for some \(i\in [r'-1]\).  The number of vectors at distance at most \(r\) from \(\vec{s}^*_i\), \(i< r'\), is less than \((1-\varepsilon')k'\).  This means that at least \(\varepsilon' k' \ge \frac{\varepsilon'|M|}{2^{r_{\max}}}\) vectors whose completions are in \(P^*\) are at distance at least \(r+1\) from \(\vec{s}^*_i\). For every such vector \(\vec{w}\), it is easy to see that, since \(\HSET(\vec{s}^*_i)\subseteq \HSET(\vec{s}^*)\), it must be the case that \((\HSET(\vec{s}^*)\cap \HSET(\vec{w}) )\setminus\HSET(\vec{s}^*_i) \) is nonempty. Note that \(|\HSET(\vec{s}^*)|\le r\), and hence there exists \(c_{i+1}\in \HSET(\vec{s}^*)\) such that, for at least \(\frac{\varepsilon'|M|}{2^{r_{\max}}\cdot r}\) vectors \(\vec{w}\) in \(M\) at distance at least \(r+1\) from \(\vec{s}^*_i\), it holds that \(c_{i+1}\in \HSET(\vec{w})\). Moreover, for every vector \(\vec{w}\in M\), we have \(|\HSET(\vec{w})|\le r_{\max}\). It follows--by  a straightforward counting argument--that there are at most \(\frac{2^{r_{\max}}\cdot r}{\varepsilon}\cdot r_{\max}\) coordinates \(c\in [d]\) such that, for at least \(\frac{\varepsilon'|M|}{2^{r_{\max}}\cdot r}\) vectors \(\vec{w}\), it holds that \(c\in \HSET(\vec{w})\). Therefore, to obtain \(\vec{s}^*_{i+1}\) such that \(\HSET(\vec{s}^*_{i+1})\subseteq \HSET(\vec{s}^*)\), we only need to branch on one of at most \(\frac{2^{r_{\max}}\cdot r}{\varepsilon}\cdot r_{\max}\) coordinates. The statement of the theorem follows since we can exhaustively branch on the coordinates that are set to 1 in at least \(\frac{\varepsilon'|M|}{2^{r_{\max}}\cdot r}\) many vectors in \(M\), until either the number of vectors at distance at most \(r\) from \(\vec{s}^*_{i}\) is at least \((1-\varepsilon')k'\) or \(i\ge r\).
\end{proof}
\fi

\iflong
The algorithm starts by performing a similar branching and pre-processing to the \FPT{} algorithm for \DIAMCq\ parameterized by \(r+\MEP(M)\). Fix \(M^*\) to be a completion of \(M\) that contains a maximum size \RCLUS{}, let \(P^*\) be such a maximum size \RCLUS{} in \(M^*\), and let \(\vec{s}^*\) be a center of \(P^*\). The goal of the algorithm is to find \(\vec{s}^*\), as given \(\vec{s}^*\), by Observation~\ref{obs:center}, we can decide the instance in time $\bigoh(|M|d)$.

If $k<\frac{2\MEP(M)}{\varepsilon}+2$, then we use the algorithm in Theorem~\ref{thm:ANY_k_r_comb} to obtain a Turing kernel\iflong; since the vectors in the Turing kernel have their dimension bounded by a function of the parameter, we can enumerate all possible centers of the \RCLUS{} to solve every sub-instance of the Turing kernel. Now we can assume that $k\ge\frac{2\MEP(M)}{\varepsilon}+2$. We first observe that $|P^*\setminus D_M|$ contains at least two vectors (recall that \(D_M\) denotes the \(\MEP(M)\)-deletion set in \(M\)).
We now\fi\ifshort . Otherwise, we can\fi{} guess two vectors \(\vec{u}\) and \(\vec{v}\) in \(M\setminus D_M\), together with their respective completions \(\vec{u}^*\) and \(\vec{v}^*\), such that \iflong both\fi{} \(\vec{u}^*\) and \(\vec{v}^*\)\iflong{} are in \(P^*\) and\fi{} are the farthest vectors apart in \(P^*\setminus D_M\); fix $r_{\max}=\HDIST(\vec{v}^*,\vec{u}^*)$.\iflong{} Note that \(r_{\max}\le 2r\). Since there are at most \(|M|^2\) pairs of vectors in \(M\setminus D_M\) and each has at most \(\MEP(M)\) missing entries, there are at most \(2^{2\MEP(M)}|M|^2\) possible pairs of  \(\vec{u}^*\) and \(\vec{v}^*\), and our algorithm enumerates all these pairs. What we describe from now on assumes that the guess of  \(\vec{u}^*\) and \(\vec{v}^*\) is correct.\fi{}
We normalize all the vectors in $M$ so that
$\vec{v}^*$ becomes the all-zero vector\iflong , \ie, we replace $\vec{v}^*$ by
the all-zero vector, and for every other vector $\vec{w}\neq \vec{v}^*$,
we replace it with the vector $\vec{w}'$ such that $\vec{w}'[i]=0$ if
$\vec{v}^*[i]=\vec{w}[i]$, $\vec{w}'[i]=\blank$ if $\vec{w}[i]=\blank$, and $\vec{w}'[i]=1$, otherwise\fi .
Finally, for each vector $\vec{w}\in M$, we can in time $r_{\max}\cdot d$, where \(d\) is the dimension of the vectors in \(M\), check if there is a completion $\vec{w}^*$ of \(\vec{w}\) such that the distance from \(\vec{w}^*\) to both \(\vec{v}^*\) and \(\vec{u}^*\) is at most \(r_{\max}\)\iflong{}. This is because we can always set \(\vec{w}^*[i]=0\) for all \(i\in [d]\setminus\HSET(\vec{u}^*)\) such that  \(\vec{w}^*[i]=\blank\), since in that case \(\vec{v}^*[i]=\vec{u}^*[i]=0\) and we are looking for a completion that is close to both \(\vec{v}^*\) and \(\vec{u}^*\). Hence, to decide if there is such completion, we only need to decide how many coordinates of \(\vec{w}^*\) in \(\HSET(\vec{u}^*)\) we need to set to \(0\) and how many we need to set to \(1\). We can now remove all vectors \(\vec{w}\) that cannot be completed to a vector of distance at most $r_{\max}$ from both $\vec{v}^*$ and $\vec{u}^*$.\fi{}\ifshort{}; we remove all vectors \(\vec{w}\) that do not have such a completion.\fi
~Note here that some vectors in \(P^*\cap D_M\)\iflong{} might be further away from $\vec{v}^*$ or $\vec{u}^*$ and they\fi{} could have been removed from \(M\) at this step.
However,\iflong{} since $k\ge\frac{2\MEP(M)}{\varepsilon}+2$, we have\fi{} \(\MEP(M)\le \frac{\varepsilon}{2}k\) and\iflong{}, if our guess of $\vec{v}^*$ and $\vec{u}^*$ is correct,\fi{} we did not remove any vector from \(P^*\setminus D_M\). Hence, after the preprocessing, \(M\) contains a \RCLUS{} with center \(\vec{s}^*\) and at least \(k'=(1-\frac{\varepsilon}{2})k\) vectors. Our goal is to find a center for a \RCLUS{} with at least \((1-\frac{\varepsilon}{2})k' = (1-\frac{\varepsilon}{2})^2k\ge (1-\varepsilon)k\) vectors. For ease of exposition, we let \(\varepsilon' = \frac{\varepsilon}{2}\) and we let \(k'\ge(1-\varepsilon')k\) be the number of vectors of \(P^*\) still in \(M\).

\begin{LEM} \label{lem:existence_of_large_rad_cluster}
	After the above pre-processing, in time \(\bigoh(2^{r_{\max}}\cdot |M|)\), we can find a center for a \RCLUS{} of size $\frac{|M|}{2^{r_{\max}}}$.
\end{LEM}
\begin{proof}
	The lemma follows rather straightforwardly  from the following claim.	\begin{CLM}
		For every vector \(\vec{w}\in M\) it holds that \(|\HSET(\vec{w})\setminus \HSET(\vec{u}^*)|\le r\).
	\end{CLM}
	\begin{claimproof}[Proof of Claim]
		Let \(\vec{w}^*\) be a completion of \(\vec{w}\) such that \(\vec{w}^*\) is at distance at most \(r_{\max}\) from both \(\vec{u}^*\) and \(\vec{v}^*\). Notice that \(\HSET(\vec{w})\setminus \HSET(\vec{u}^*)\subseteq \HSET(\vec{w}^*)\setminus \HSET(\vec{u}^*)\); hence \(|\HSET(\vec{w})\setminus \HSET(\vec{u}^*)|\le |\HSET(\vec{w}^*)\setminus \HSET(\vec{u}^*)|\), and it suffices to show that \(|\HSET(\vec{w}^*)\setminus \HSET(\vec{u}^*)|\le r\).
		
		Now \(\vec{w}^*\) is at distance at most \(r_{\max}\) from \(\vec{v}^*\), and hence we have \(\HDIST(\vec{w}^*,\vec{v}^*)\le r_{\max}\) and \(\HSET(\vec{v}^*)=\emptyset\). It follows that  \(|\HSET(\vec{w}^*)\setminus \HSET(\vec{u}^*)| + |\HSET(\vec{w}^*)\cap \HSET(\vec{u}^*)|\le r_{\max}.\) Similarly, \(\vec{w}^*\) is at distance at most \(r_{\max}\) from \(\vec{u}^*\), i.e., \(\HDIST(\vec{w}^*,\vec{u}^*)\le r_{\max}\), and  \(|\HSET(\vec{w}^*)\setminus \HSET(\vec{u}^*)| + |\HSET(\vec{u}^*)\setminus \HSET(\vec{w}^*)|\le r_{\max}.\) Combining the two inequalities, we get \(2|\HSET(\vec{w}^*)\setminus \HSET(\vec{u}^*)| + |\HSET(\vec{u}^*)\cap \HSET(\vec{w}^*)| + |\HSET(\vec{u}^*)\setminus \HSET(\vec{w}^*)|\le 2r_{\max}.\) Since \(|\HSET(\vec{u}^*)\cap \HSET(\vec{w}^*)| + |\HSET(\vec{u}^*)\setminus \HSET(\vec{w})| = |\HDIST(\vec{v}^*,\vec{u}^*)| = r_{\max},\) we have  \(|\HSET(\vec{w}^*)\setminus \HSET(\vec{u}^*)|\le \frac{r_{\max}}{2}\le r.\)
	\end{claimproof} 
We can now branch over all possible subsets \(T\) of \(\HSET(\vec{u}^*)\), and for each such subset, try as the center the vector \(\vec{s}^*_T\) such that \(\HSET(\vec{s}^*_T)=T\) and compute the set of all vectors at distance at most \(r\) from \(\vec{s}^*_T\). To prove the lemma, it remains to show that every vector in \(M\) can be completed to a vector at distance at most \(r\) from at least one of \(2^{r_{\max}}\) vectors \(\vec{s}^*_T\).
Let \(\vec{w}\) be a vector in \(M\) that can be completed to a vector \(\vec{w}^*\)  with \(\HSET(\vec{w}^*)\cap \HSET(\vec{u}^*) = T\) for some \(T\subseteq \HSET(\vec{u}^*)\) . Now let \(\vec{w}^*\) be such completion of \(\vec{w}\) that \(\HSET(\vec{w}^*)\cap \HSET(\vec{u}^*) = T\) and \(\vec{w}^*[i]=0\) for all \(i\in [d]\setminus \HSET(\vec{u}^*)\) such that \(\vec{w}[i]=\blank\). It follows from the claim that \(\vec{w}^*\) is at distance at most \(r\) from \(\vec{s}^*_T\).
\end{proof}

Therefore, we can assume henceforth that \(k'\ge \frac{|M|}{2^{r_{\max}}}\)\iflong ; otherwise, Lemma~\ref{lem:existence_of_large_rad_cluster} computes a \RCLUS{} of size at least \(k'\)\fi .
Now the algorithm sets \(\vec{s}^*_0 = \vec{v}^* = \vec{0}\) and the goal is to iteratively compute
\(\vec{s}^*_1, \vec{s}^*_2, \vec{s}^*_3, \ldots, \vec{s}^*_{r'}\), \(r'\le r\), such that:

\begin{enumerate}
	\item for all \(i\in [r']\), we have	\(\HSET(\vec{s}^*_i) = \HSET(\vec{s}^*_{i-1})\cup \{c_i\}\) for some coordinate \(c_i\in \HSET(\vec{s}^*)\setminus \HSET(\vec{s}^*_{i-1})\);
	\item for all $j\in[r'-1]$, the number of vectors at distance at most \(r\) from \(\vec{s}^*_{j}\) is less than \((1-\varepsilon')k'\); and
	\item the number of vectors at distance at most \(r\) from \(\vec{s}^*_{r'}\) is at least \((1-\varepsilon')k'\).
\end{enumerate}

Let \(\vec{s}^*_i\) be such that \(\HSET(\vec{s}^*_i)\subseteq \HSET(\vec{s}^*)\) for some \(i\in [r'-1]\).  The number of vectors at distance at most \(r\) from \(\vec{s}^*_i\), \(i< r'\), is less than \((1-\varepsilon')k'\).  This means that at least \(\varepsilon' k' \ge \frac{\varepsilon'|M|}{2^{r_{\max}}}\) vectors whose completions are in \(P^*\) are at distance at least \(r+1\) from \(\vec{s}^*_i\). For every such vector \(\vec{w}\), it is easy to see that since \(\HSET(\vec{s}^*_i)\subseteq \HSET(\vec{s}^*)\), it must be the case that \((\HSET(\vec{s}^*)\cap \HSET(\vec{w}) )\setminus\HSET(\vec{s}^*_i) \) is not empty. Note that \(|\HSET(\vec{s}^*)|\le r\), and hence there exists \(c_{i+1}\in \HSET(\vec{s}^*)\) such that, for at least \(\frac{\varepsilon'|M|}{2^{r_{\max}}\cdot r}\) vectors \(\vec{w}\) in \(M\) at distance at least \(r+1\) from \(\vec{s}^*_i\), it holds that \(c_{i+1}\in \HSET(\vec{w})\). Moreover, for every vector \(\vec{w}\in M\), we have \(|\HSET(\vec{w})|\le r_{\max}\). It follows--by  a straightforward counting argument--that there are at most \(\frac{2^{r_{\max}}\cdot r}{\varepsilon}\cdot r_{\max}\) coordinates \(c\in [d]\) such that, for at least \(\frac{\varepsilon'|M|}{2^{r_{\max}}\cdot r}\) vectors \(\vec{w}\), it holds that \(c\in \HSET(\vec{w})\). Therefore, to obtain \(\vec{s}^*_{i+1}\) such that \(\HSET(\vec{s}^*_{i+1})\subseteq \HSET(\vec{s}^*)\), we only need to branch on one of at most \(\frac{2^{r_{\max}}\cdot r}{\varepsilon}\cdot r_{\max}\) coordinates. By exhaustively branching on the coordinates that are set to 1 in at least \(\frac{\varepsilon'|M|}{2^{r_{\max}}\cdot r}\) many vectors in \(M\), until either the number of vectors at distance at most \(r\) from \(\vec{s}^*_{i}\) is at least \((1-\varepsilon')k'\) or \(i\ge r\),  we get:

\begin{THE} \label{the:radius-approximation}
		Given an instance \((M,k,r)\) of \RADCq\ and \(\varepsilon\in \mathbb{R}\), where \(0<\varepsilon<1\), there exists an \FPT{}
		algorithm \(\mathcal{A}\), parameterized by $r+\MEP(M)+\frac{1}{\varepsilon}$, such that \(\mathcal{A}\) either computes a \RCLUS{} of size at least $(1-\varepsilon)k$, or correctly concludes that \(M\) does not contain a \RCLUS{} of size \(k\).
	\end{THE}

\begin{proof}
	The initial phase of the algorithm either computes a \RCLUS{} of size at least \(k\) if \(k<\frac{2\MEP(M)}{\varepsilon}+2\), or branches into at most \(2^{2\MEP(M)}|M|^2\) many branches. Afterwards, the algorithm pre-processes every branch such that in each branch we can either compute a center for a \RCLUS{} of size at least \(k'\ge k-\MEP(M)\ge (1-\frac{\varepsilon}{2})k\) in time \(\bigoh(2^{r_{\max}}|M|)\) by Lemma~\ref{lem:existence_of_large_rad_cluster}, or we end up in an instance \((M',k',r)\) such that
		$k'\ge \frac{|M|}{2^{r_{\max}}}$ and
		\(|\HSET(\vec{w})|\le r_{\max}\) for all \(\vec{w}\in M\).
Moreover, we can assume that there exists a maximum size cluster \RCLUS{} \(P^*\) in \(M'\) with \(\vec{0}\in P^*\). Note that \(k'\ge (1-\frac{\varepsilon}{2})k\) and a center of \RCLUS{} of size \(k'\) in \(M'\) can be straightforwardly transformed to a center of \RCLUS{} of size \(k'\) in \(M\).

From now on the algorithm is rather simple. It starts with \(s^*_0= \vec{0}\). Now for each \(i\in \{0,\ldots,r\}\) it either concludes that \(s^*_i\) is a good enough center or it branches in at most \(\frac{2^{r_{\max}}\cdot r}{\varepsilon}\cdot r_{\max}\) possibilities for \(s^*_{i+1}\) with \(|\HSET(s^*_{i+1})|=i+1\) as follows.
\begin{enumerate}
	\item It first computes the set \(P_i\) of all vectors in \(M\) at distance at most \(r\) from \(\vec{s}^*_i\). This can be done in \(\bigoh(|M|d)\) time, where \(d\) is the dimension of the vectors in \(M\).
	\item If this set has size at least \((1-\varepsilon')k'\), it outputs as the center the vector \(\vec{s}^*_i\) and stops.
	\item Now, if \(i = r\) and we did not find a good center, then any center containing \(\vec{s}^*_r\) would be too far from \(\vec{0}\), so the algorithm outputs \textsc{Fail} and stops.
	
	\item Otherwise it branches over all coordinates \(c\in [d]\setminus \HSET(\vec{s}^*_i)\) it computes the set \(W\subseteq (M\setminus P_i)\) of all vectors \(\vec{w}\) with distance at least \(r+1\) from \(\vec{s}^*_i\) and \(c\in \HSET(\vec{w})\).
	\begin{enumerate}
		\item If \(|W|\ge \frac{\varepsilon'|M|}{2^{r_{\max}}\cdot r}\), then the algorithm branches on letting \(\vec{s}^*_{i+1}\) be the vector such that \(\vec{s}^*_{i+1}[c]=1\) and \(\vec{s}^*_{i+1}[j]=\vec{s}^*_{i}[j]\) for all \(j\neq c_i\).
		\item If \(|W| < \frac{\varepsilon'|M|}{2^{r_{\max}}\cdot r}\), discard the current choice of $c$.
	\end{enumerate}
\end{enumerate}

The above algorithm is a branching algorithm, where the maximum number of branches at each node is \(\frac{2^{r_{\max}}\cdot r}{\varepsilon}\cdot r_{\max}\)
and the branching depth is \(r\). Otherwise, the algorithm runs in polynomial time in every branching node. Hence the total running time of the algorithm is
\((\frac{2^{r_{\max}}\cdot r}{\varepsilon}\cdot r_{\max})^{r}\cdot |M|^{\bigoh(1)}\). If the algorithm outputs a vector \(\vec{s}^*_i\), then there are at least
\((1-\varepsilon')k'\) vectors in \(M'\) at distance at most \(r\) from \(\vec{s}^*_i\). On the other hand, if there is \RCLUS{} \(P^*\) of size at
least \(k'\) with center \(\vec{s}^*\) such that \(\vec{0}\in P^*\), then
\(\HSET(\vec{s}^*_0)=\emptyset\subseteq \HSET(\vec{s}^*)\),
\(|\HSET(\vec{s}^*)|\ge r\), and for every vector \(\vec{s}^*_i\) with
\(\HSET(\vec{s}^*_i)\subseteq\HSET(\vec{s}^*)\) we already argued that there
exists \(c_{i+1}\in \HSET(\vec{s}^*)\setminus \HSET(\vec{s}^*_i)\) such that for at least \(\frac{\varepsilon'|M|}{2^{r_{\max}}\cdot r}\) vectors \(\vec{w}\) in \(M\) at distance at least \(r+1\) from \(\vec{s}^*_i\) it holds \(c_{i+1}\in \HSET(\vec{w})\). Hence the node with \(\vec{s}^*_i\) contains a branch with center \(\vec{s}^*_{i+1}\) such that \(\HSET(\vec{s}^*_{i+1})\subseteq \HSET(\vec{s}^*)\). If \(\HSET(\vec{s}^*_{i})\subseteq \HSET(\vec{s}^*)\) and \(i = |\HSET(\vec{s}^*)|\), then \(\vec{s}^*_i = \vec{s}^*\) and there are \(k'\) vectors at distance at most \(r\) from \(\vec{s}^*_i\). It follows that if \((M',k',r)\) is a \yes-instance, then the above algorithm indeed outputs some vector \(\vec{s}^*_i\) with at least \((1-\varepsilon')k'\) vectors from \(M'\) at distance at most \(r\) from \(\vec{s}^*_i\).
\end{proof} \fi

\section{Concluding Remarks}

We studied the parameterized complexity of two fundamental problems pertaining to incomplete data that have applications in data analytics. In most cases, we were able to provide a complete landscape of the parameterized complexity of the problems w.r.t.~the parameters under consideration. 
It is worth noting that all algorithmic upper bounds obtained in this paper can also be directly generalized to vectors (i.e., matrices) over a domain whose size is bounded by the parameter value by using the encoding described by Eiben et al.~\cite{EibenGKOS21}.

Two important open questions ensue from our work, namely determining the parameterized complexity of \RADCq{} w.r.t.~each of the two parameterizations $k + \lambda$ and $r + \lambda$. In particular, the restrictions of these two problems to complete data (i.e., $\lambda =0$) remain open, and result in two important questions about the parameterized complexity of \RADC{} parameterized by the cluster size $k$ or by the cluster radius $r$.

\bibliography{literature}

\end{document}